\PassOptionsToPackage{table}{xcolor}

\documentclass[11pt]{article}
\usepackage{arxiv}
\usepackage{natbib}
\usepackage{hyperref}
\let\oldthebibliography\thebibliography  
\renewcommand{\thebibliography}[1]{  
  \oldthebibliography{#1}  
  \normalsize %
}

\usepackage{tikz}

\usepackage{fixmath}

\usepackage[utf8]{inputenc}

\usepackage{caption}

\usepackage{tikz}
\usetikzlibrary{shapes, arrows.meta, positioning}

\usepackage{amsmath}
\usepackage{dsfont}
\usepackage{amssymb,bbm,mathtools}
\usepackage{amsthm}
\usepackage{pifont}%

\usepackage{microtype} %

\usepackage{blindtext}

\usepackage{bm} %
\usepackage{graphicx}
\usepackage{url}
\usepackage{booktabs} %
\usepackage{threeparttable}
\usepackage{enumitem} %
\usepackage[table]{xcolor}
\usepackage{stmaryrd}
\usepackage{bm}
\usepackage{breqn} %
\usepackage{mdwlist} %
\usepackage{algorithm}
\usepackage{algorithmic}
\usepackage{ifthen}
\newboolean{mybool}
\setboolean{mybool}{true} %

\definecolor{burntorange}{rgb}{0.8, 0.33, 0.0}

\def\isanonymous{0}
\def\isfullversion{0}

\newcommand{\ifanon}[2]{
\ifthenelse{\equal{\isanonymous}{1}}
{{#1}}
{{#2}}
}

\newcommand{\fullversion}[2]{
\ifthenelse{\equal{\isfullversion}{1}}{{#1}}{{#2}}}

\definecolor{dartmouthgreen}{rgb}{0.05, 0.5, 0.06}

\newcommand{\ours}{SLVR}

\newcommand{\R}{\mathbb{R}}

\mathchardef\mhyphen="2D

\newcommand{\msf}[1]{\mathsf{#1}}
\newcommand{\mypara}[1]{\setlength{\parskip}{5pt}
\noindent\textbf{#1}}

\newcommand{\ie}{\textit{i.e.,}\@\xspace}
\newcommand{\eg}{\textit{e.g.,}\@\xspace}

\newcommand{\makevect}[1]{{\ensuremath{\mathbf{#1}}}}

\renewcommand{\vec}[1]{\makevect{#1}}

\newcommand{\F}{\mathbb{F}}

\newcommand{\PartySet}[1]{\mathcal{P}_{#1}}

\newcommand{\server}{\mathcal{S}}
\newcommand{\mpcnodes}{\PartySet{MPC}}

\newcommand{\cor}{m_c}

\newcommand{\valid}{\msf{Valid}}

\newcommand{\arith}[1]{\llbracket #1\rrbracket}

\newcommand{\scr}{\msf{Chk}}
\newcommand{\score}[2]{\msf{scr}^{#2}_{#1}}
\newcommand{\Acc}[2]{\msf{Acc}^{#2}_{#1}}
\newcommand{\diff}[2]{\msf{diff}^{#2}_{#1}}

\newcommand{\Sim}{\mathcal{S}}
\newcommand{\Adv}{\mathcal{A}}

\newcommand{\Prot}[1]{\ensuremath{\Pi_{\mathsf{#1}}}}
\newcommand{\PSecAgg}{\Prot{Sec \mhyphen Agg}}
\newcommand{\PShare}{\Prot{Share}}
\newcommand{\PRec}{\Prot{Recon}}
\newcommand{\PCheck}{\Prot{Check}}
\newcommand{\PMult}{\Prot{Mult}}
\newcommand{\PSqrt}{\Prot{Sqrt}}

\newcommand{\Func}[1]{\ensuremath{\mathcal{F}_{\mathsf{#1}}}}
\newcommand{\FSecAgg}{\Func{Sec \mhyphen Agg}}
\newcommand{\FSort}{\Func{Sort}}
\newcommand{\FSecInf}{\Func{Sec \mhyphen Inf}}
\newcommand{\FRand}{\Func{Rand}}
\newcommand{\FMaxSoft}{\Func{Max \mhyphen Soft}}
\newcommand{\FZeroOne}{\Func{Zero \mhyphen One}}
\newcommand{\FMult}{\Func{Mult}}
\newcommand{\FComp}{\Func{Comp}}

\newcommand{\appref}[1]{Appendix~\ref{app:#1}}

\newcommand{\figref}[1]{Fig.~\ref{fig:#1}}
\newcommand{\tabref}[1]{Table~\ref{tab:#1}}
\newcommand{\secref}[1]{Section~\ref{sec:#1}}

\newcommand{\figlab}[1]{\label{fig:#1}}

\usepackage{tcolorbox}
\tcbuselibrary{skins,breakable}

\newtcolorbox{protobox}[2][]{%
  enhanced,
  title        = {#2},
  attach boxed title to top left={xshift=+3mm,yshift*=-3mm},
  breakable    = false,
  colback      = white, %
  colframe     = black!75,
  fonttitle    = \bfseries,
  colbacktitle = black!10!white,
  coltitle     = black,
  #1
}

\newenvironment{protofig}[3]{%
  \begin{figure}[!h]
    \newcommand{\FigCaption}{#2}
    \newcommand{\FigLabel}{#3}
    \begin{protobox}{#1}
    }{%
    \end{protobox}
    \vspace{-1em}
    \caption{\FigCaption}
    \figlab{\FigLabel}
  \end{figure}
}

\theoremstyle{plain}
\newtheorem{theorem}{Theorem}

\theoremstyle{definition}
\newtheorem{definition}[theorem]{Definition}

\theoremstyle{remark}

\newcommand{\argmax}{\operatornamewithlimits{arg\!\max}}
\newcommand{\ben}{\begin{enumerate}}
\newcommand{\een}{\end{enumerate}}
\newcommand{\beq}{\begin{equation}}
\newcommand{\eeq}{\end{equation}}
\newcommand{\beqa}{\begin{eqnarray}}
\newcommand{\eeqa}{\end{eqnarray}}
\newcommand{\bit}{\begin{itemize}}
\newcommand{\eit}{\end{itemize}}
\newcommand{\btab}{\begin{tabular}}
\newcommand{\etab}{\end{tabular}}

\newcommand{\noprint}[1]{}

\renewcommand{\thesubsubsection}{\arabic{section}.\arabic{subsubsection}}

\def \ie {{\em i.e.},~}

\def \eg {{\em e.g.},~}

\newcommand{\indicator}[1]{\mathds{1}\left[#1\right]}

\def \bff {\mathbf{f}}

\def \bfs {\mathbf{s}}

\def \bfu {\mathbf{u}}

\def \bfw {\mathbf{w}}
\def \bfx {\mathbf{x}}

\def \bfD {\mathbf{D}}

\def \calA {\mathcal{A}}

\def \calC {\mathcal{C}}
\def \calD {\mathcal{D}}

\def \calL {\mathcal{L}}

\def \calS {\mathcal{S}}

\def \calV {\mathcal{V}}

\usepackage{adjustbox}
\usepackage{longtable}
\usepackage{multirow}
\usepackage{subcaption}
\usepackage{enumitem}
\usepackage{fancyvrb}
\usepackage{dsfont}
\usepackage{changepage}

\usepackage{booktabs} %
\usepackage{xspace}

\usepackage{calc} %

\begin{document}

\title{\ours: Securely Leveraging Client Validation \\ for Robust Federated Learning}
\author{
Jihye Choi$\hspace{.5mm}^1$\thanks{Work done while at Visa Research.}~~,   
Sai Rahul Rachuri$\hspace{.5mm}^2$, 
Ke Wang$\hspace{.5mm}^{2 *}$, 
Somesh Jha$\hspace{.5mm}^1$,
Yizhen Wang$\hspace{.5mm}^2$ \\
$^1\hspace{.5mm}$University of Wisconsin-Madison ~ $^2\hspace{.5mm}$Visa Research}

\newcommand{\fix}{\marginpar{FIX}}
\newcommand{\new}{\marginpar{NEW}}

\let\oldcite\cite
\renewcommand{\cite}{\citep}

\maketitle

\begin{abstract}
    
    Federated Learning (FL) enables collaborative model training while keeping client data private. However, exposing individual client updates makes FL vulnerable to reconstruction attacks. 
    Secure aggregation mitigates such privacy risks but prevents the server from verifying the validity of each client update, creating a privacy-robustness tradeoff.
    Recent efforts attempt to address this tradeoff by enforcing checks on client updates using zero-knowledge proofs, but they support limited predicates and often depend on public validation data.
    We propose \ours, a general framework that securely leverages clients' private data through secure multi-party computation.
    By utilizing clients' data, \ours\ not only eliminates the need for public validation data, but also enables a wider range of checks for robustness, including cross-client accuracy validation.
    It also adapts naturally to distribution shifts in client data as it can securely refresh its validation data up-to-date. 
    Our empirical evaluations show that \ours\ improves robustness against model poisoning attacks, particularly outperforming existing methods by up to 50\% under adaptive attacks. Additionally, \ours\ demonstrates effective adaptability and stable convergence under various distribution shift scenarios.

\end{abstract}

\section{Introduction}
\label{sec:intro}

Federated Learning (FL) is a widely used paradigm for training models across distributed data sources~\cite{fedavg}. 
In FL, clients compute model updates locally on private data and send them to a central server. 
FL was originally designed with privacy in mind, aiming to share model updates instead of raw data, under the assumption that these updates would not leak sensitive information. However, recent attacks have shown that exposing client updates in the clear makes FL vulnerable to full-scale reconstruction attacks~\citep{zhu2019deep,geiping2020inverting}.

In response to these privacy risks, FL has been combined with secure aggregation techniques such as multiparty computation (MPC)~\cite{bonawitz2017practical, bell2020secure, fereidooni2021safelearn, so2022lightsecagg} or zero-knowledge proofs (ZKPs) to ensure that the server only learns the aggregated update and not individual updates (Figure~\ref{fig:a_secure}).
However, this privacy guarantee comes at the cost of robustness --- since the server no longer sees individual updates, it lacks a mechanism to verify their validity before aggregation. 
There is active research on Byzantine-robust defense mechanisms to filter out malformed updates in plaintext FL~\cite{sun2019normbound, steinhardt2017normball, yin2018trimmedmean, blanchard2017krum, cao2020fltrust}, but integrating these checks into secure aggregation is challenging.
Since these checks are typically built in plaintext, it is generally the case that the server has access to data, or metadata about the updates that it would not have when using secure aggregation. This makes it difficult to translate these predicates to secure aggregation.

\begin{figure*}[t]
\centering

\begin{subfigure}{0.32\textwidth}
\includegraphics[width=\linewidth]{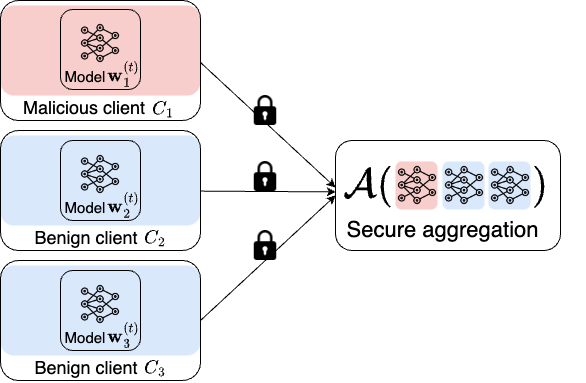}
\caption{Secure aggregation without robustness against malicious clients updates.}
\label{fig:a_secure}
\end{subfigure}
\hspace{.1em}
\begin{subfigure}{0.32\textwidth}
\includegraphics[width=\linewidth]{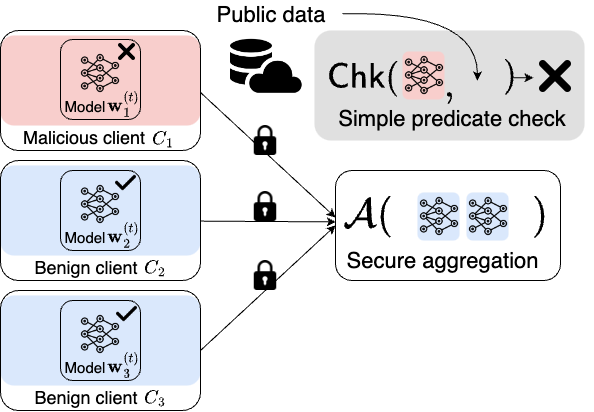}
\caption{Robust secure aggregation with simple model integrity check on public data.}
\label{fig:b_robust_with_publicval}
\end{subfigure}
\hspace{.1em}
\begin{subfigure}{0.32\textwidth}
\includegraphics[width=\linewidth]{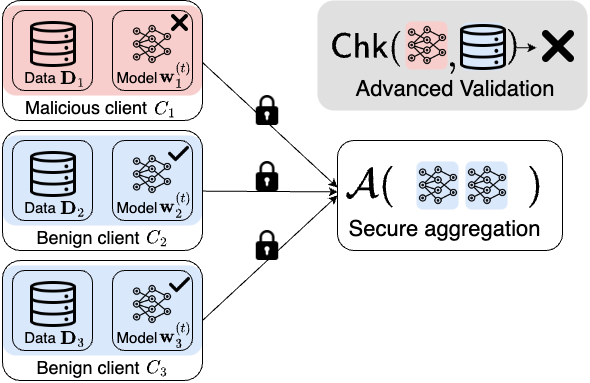}
\caption{\ours~(ours) with advanced validation check securely leveraging client private data.}
\label{fig:c_ours}
\end{subfigure}

\caption{ \textbf{Previous methods vs. \ours.}
Secure aggregation~(\ref{fig:a_secure}) lacks the means to check model integrity privately. Recent byzantine-robust secure aggregation schemes enable integrity checks with zero-knowledge proofs~(\ref{fig:b_robust_with_publicval}), but only support simple predicate (\eg $\ell_2$ norm of client update), and may rely on public datasets. ~\ours~securely leverages clients' local data for integrity checks via MPC~(\ref{fig:c_ours}), supporting more powerful predicates without requiring public data. 
}

\label{fig:ours-overview}
\end{figure*}

There are recent efforts to incorporate robustness checks into secure FL~\cite{franzese2023robust}, but they remain limited in terms of the flexibility of checks that are supported or the reliance on public validation data (Figure~\ref{fig:b_robust_with_publicval}).
Specifically, ZKP-based approaches~\cite{rofl, elsa} only support simple predicates, such as $l_2$ and $l_{\infty}$ norm bounds, while Eiffel~\cite{CCS:CGJv22} extends support to more predicates but relies on a public data to determine the predicate hyperparameters.
In practice, public validation data may not always be available, and even when it exists, its utility depends on how well it represents the overall client distribution. As client data naturally evolves or new clients joining over time, maintaining an effective public dataset becomes non-trivial, particularly in secure FL setups. 

To this end, we attempt to resolve the dilemma of privacy vs. robustness of FL from a fresh angle by asking:
\begin{tcolorbox}
    \begin{itemize}[leftmargin=*]
        \item Can we leverage clients' private data while preserving the privacy guarantee of secure aggregation?
        \item Can leveraging private data enable more advanced checks that enhance robustness while overcome the reliance on public validation data? 
    \end{itemize}
\end{tcolorbox}
We present~\ours, a general framework for \textbf{S}ecurely \textbf{L}everaging clients' private data to \textbf{V}alidate updates for \textbf{R}obust FL.
Aided by MPC, \ours\ can compute statistics over multiple clients' private data that is otherwise impossible in previous secure FL work (Figure~\ref{fig:c_ours}). These statistics, such as cross-client validation accuracy, can play a vital role in enhancing model robustness.
The modular design of \ours\ makes it compatible with existing MPC implementations, and thus allows flexible deployment in various threat models.
Furthermore, \ours\ naturally adapts to distribution shifts as it can securely refresh its validation data up-to-date.  

In summary, our contributions are as follows,
\begin{enumerate}
    \item We propose the first framework that enables cross-client checks using private inputs from different clients, eliminating reliance on public validation data. Our approach is compatible with existing MPC protocols, allowing flexible customization based on the threat model and computational constraints.
    \item  We evaluate \ours~against model poisoning attacks, including an adaptive attack---the strongest possible under the threat model. \ours~demonstrates competitive robustness and outperforms prior work under adaptive attacks (\eg by up to 50\% on CIFAR-10).
    \item We study various scenarios for client data distribution shift scenarios and empirically demonstrate the adaptability of \ours. In contrast, prior works struggle, experiencing severe accuracy degradation (\eg 30\% drop from MNIST to SVHN) or failing to progress.
\end{enumerate}

\section{Related Work}
\label{sec:related}

\subsection{Robust Federated Learning}

\mypara{Robustness to Model Poisoning.} Model poisoning is one of the most fundamental threats to FL. As formulated by~\citet{blanchard2017byzantine} and~\citet{yin2018trimmedmean}, a set of Byzantine malicious clients can send manipulated model updates to the server to degrade the global model's performance.
To mitigate this, various defenses have been proposed, primarily through robust aggregation methods that aim to filter out malicious updates from the final aggregation~\cite{sun2019normbound, steinhardt2017normball, yin2018trimmedmean, blanchard2017krum, cao2020fltrust}.
However, these approaches often fall short against adaptive attacks~\cite{fang2020local,SP:LBVKH23} or require access to a public dataset for validation~\cite{cao2020fltrust}, which may not always be available.

\mypara{Robustness to Distribution Shift.}
Another key challenge in FL is handling evolving client data distribution over communication rounds. 
Many approaches address this from the perspective of client heterogeneity, incorporating techniques such as transfer learning, multi-task learning, and meta-learning to mitigate distribution shifts across different domains and tasks~\cite{smith2017federated, khodak2019adaptive}.
More recent works study explicitly tackles evolving client distributions over time~\cite{yoon2021federated} or periodic client distribution shifts~\cite{periodicshift}.

While these advances improve the robustness of FL, they are designed for standard (plaintext) FL, where model updates are exchanged in the clear. 
However, privacy-preserving federated learning (PPFL) introduces additional constraints, making it non-trivial to directly apply these techniques. 
Many existing robust FL methods rely on inspecting individual model updates, a process that is inherently difficult in PPFL due to encryption or secure aggregation mechanisms. 
Our work aims to bridge this gap by providing a unified framework for PPFL that enhances robustness against both model poisoning attacks and distribution shifts.

\subsection{Privacy-preserving Federated Learning}
Various techniques have been proposed for PPFL, with key differences in their underlying security assumptions.
One of the primary differentiators is whether a protocol relies on a single-server or a multi-server aggregator. 
Single-server protocols for secure aggregation include works like \cite{CCS:BBGLR20,SP:MWAPR23,rofl,CCS:CGJv22}, while \cite{SP:RSWP23,scionfl} represents some of the state-of-the-art approaches in the multi-server setting. These techniques typically involve one or more of the following methods: secret-sharing \cite{SP:MWAPR23}, homomorphic encryption (HE) \cite{CCS:BIKMMP17}, and zero-knowledge (ZK) \cite{rofl,CCS:CGJv22}.

However, these methods generally focus on ensuring privacy and do not protect against adversarial behaviors such as model poisoning, as they do not validate client-submitted updates before aggregation.
To address this, recent research has explored incorporating zero-knowledge proofs (ZKPs) to enable privacy-preserving validation of model updates. 
While initial approaches were computationally expensive, advancements in ZK SNARKs, such as Bulletproofs~\cite{SP:BBBPWM18}, have significantly improved efficiency. 
ZK SNARKs offers small proof sizes and fast verification, making them particularly well-suited for FL, where clients often have bandwidth and computational constraints.

\begin{table}[!t]
\centering
\begin{adjustbox}{width=\columnwidth,center}
\begin{tabular}{lcccccc}
\hline
\rowcolor{gray!30}
\bf Protocol & 
\textbf{Trust Model} & 
\textbf{Input Validation} & 
\textbf{Private Bound} & 
\textbf{Dynamic Bound} &
\textbf{Single/Multi-Input Validation} & 
\textbf{Distribution Shift} \\ 
\hline
Clear & $\circ$ & \textcolor{dartmouthgreen}{\checkmark} & \textcolor{red}{$\times$} & \textcolor{red}{$\times$} & Multi & \textcolor{dartmouthgreen}{\checkmark} \\ \hline
RoFL~\cite{SP:LBVKH23} & $\circ$ & \textcolor{dartmouthgreen}{\checkmark} & \textcolor{red}{$\times$} & \textcolor{red}{$\times$} & Single & \textcolor{red}{$\times$} \\ \hline
EIFFeL~\cite{CCS:CGJv22} & $\circ$ & \textcolor{dartmouthgreen}{\checkmark} & \textcolor{red}{$\times$} & \textcolor{red}{$\times$} & Single & \textcolor{red}{$\times$} \\ \hline
Elsa~\cite{SP:RSWP23} & $\circ \bullet$ & \textcolor{dartmouthgreen}{\checkmark} & \textcolor{red}{$\times$} & \textcolor{red}{$\times$} & Single & \textcolor{red}{$\times$} \\ \hline
Mario~\cite{mario} & $\circ \bullet$ & \textcolor{dartmouthgreen}{\checkmark} & \textcolor{red}{$\times$} & \textcolor{red}{$\times$} & Single & \textcolor{red}{$\times$} \\ \hline
\rowcolor{yellow!30}
\textbf{SLVR (Ours)} & $\circ \bullet$ & \textcolor{dartmouthgreen}{\checkmark} & \textcolor{dartmouthgreen}{\checkmark} & \textcolor{dartmouthgreen}{\checkmark} & Multi & \textcolor{dartmouthgreen}{\checkmark} \\ \hline
\end{tabular}
\end{adjustbox}
\caption{\small Comparison of protocols based on trust model, input validation, bound computation, single/multi-input validation, and support for data distribution shifts. $\circ$ represents a single-aggregator trust model (without privacy in the case of Clear), while $\circ \bullet$ represents distributed trust with two or more entities. 
}
\label{tab:new_comparison}
\end{table}

\subsection{Robust and Privacy-preserving Federated Learning}

Among existing approaches that integrate robustness and privacy in FL, RoFL~\cite{rofl}, EIFFeL~\cite{CCS:CGJv22}, and ELSA~\cite{elsa} are the most relevant to our work.
They aim to enhance the robustness of PPFL under various threat models while integrating secure aggregation. 
A common strategy in these works is to ensure that only valid model updates contribute to the global model by using predicate-based validation, where clients prove the validity of their updates through ZKPs.

In contrast, our framework takes a different approach by leveraging secure multiparty computation (MPC) to compute useful robustness-related statistics on secret-shared data. 
This design allows for greater flexibility in defining validation checks, as opposed to relying solely on predefined ZKP predicates. 
In addition, \ours~allows the predicate to take multiple clients' private data as inputs, which enables cross-client validation that is neither feasible in plaintext FL nor in previous works. 
We compare with the features provided by some of the state-of-the-art frameworks in \tabref{new_comparison}.

\mypara{Secure Multiparty Computation (MPC).} 
Secure MPC is a well-established cryptographic technique that enables multiple parties to jointly compute a function over their private inputs while ensuring that no individual party learns anything beyond the output.
There are a range of applications for which MPC could be applied, such as secure auctions~\cite{FC:BCDGJK09}, privacy-preserving data analytics~\cite{prio,FC:BogTalWil12}, privacy-preserving secure training and inference~\cite{C:EGKRS20,CCS:MohRin18}, to name a few. 
In this work, the two building blocks we will use are protocols for secure inference and randomness generation.

\section{Problem Overview}
\label{sec:prelim}

In this section, we introduce the FL setting (Section~\ref{sec:setup}), followed by its threat analysis (Section~\ref{sec:sec-model}) and an overview of our goal (Section~\ref{sec:goals}).
A comprehensive list of notataions used throughout the paper can be found in Table~\ref{tab:notations}.

\subsection{Federated Learning with Secure Aggregation}
\label{sec:setup}

Federated learning with secure aggregation is a recent advancement in FL that attempts to enhance the security and privacy of FL. The clients no longer send plain-text model updates to the server for aggregation. Instead, the updates are aggregated via a secure computation protocol. We assume the existence of an \emph{aggregation committee} to facilitate the secure aggregation. In some scenarios, such as peer-to-peer (P2P) learning, the aggregation committee could be sampled from the set of clients \cite{SP:MWAPR23,CCS:BBGLR20}, whereas in the multi-server federated learning setting, the aggregation committee could just be the servers. We abstract these differences away by assuming that there exists a set of MPC nodes, denoted by $\mpcnodes$, that receive the updates from the clients and perform secure aggregation.

Formally, we consider an FL setting in which $m$ clients collaboratively train a global model $\bff_{\bfw}(\cdot)$ maintained on a single cloud server, $\calS$, parameterized by $\bfw$. Each client $\calC_i$ has its own local training dataset $\bfD_i = \{(\bfx_j, y_j)\}_{j=1}^{N_i}$ where $N_i = |\bfD_i|$, for $i\in [m]$. 
At the start of $t$-th communication round, the server broadcasts the current global model parameters $\bfw^{(t)}$ to all the clients and the aggregation committee (of which the server may be a part of). Each client $\calC_i$ uses $\bfw^{(t)}$ as the initial model, locally computes an update $\bfu^{(t)}_{i}$ on its local data $\bfD_i$ and then sends $\bfu^{(t)}_i$ to $\mpcnodes$ in secret-shared form. 
$\mpcnodes$ run the validation protocol, aggregates using an aggregation function $\calA$ (\eg weighted sum over the updates~\cite{fedavg}), and reconstructs the aggregated update $\bfu_{\msf{aggr}}^{(t)} = \calA(\bfu^{(t)}_1, \bfu^{(t)}_2, \dots, \bfu^{(t)}_{m})$ to $\calS$. 
$\calS$ updates the global model $\bfw^{(t+1)}$ based on $\bfu_{\msf{aggr}}^{(t)}$~\footnote{We refer to client model $\bfw^{(t)}_i$ and update $\bfu^{(t)}_i$ interchangeably because $\bfu^{(t)}_i = \bfw^{(t)}_i-\bfw^{(t-1)}$ and the global model $\bfw^{(t-1)}$ is public knowledge.}.
This process is repeated for $T$ communication rounds.

\mypara{MPC.} An MPC protocol enables a set of $r$ parties, denoted by $\PartySet{} = \{ P_1, \ldots, P_r \}$, to jointly compute a function $f$ that they all agree upon, on their private inputs $x_1, \ldots, x_r$. MPC guarantees that the parties only ever learn the final result of the computation, $f(x_1, \ldots, x_r)$, and nothing else about the inputs is leaked by any intermediate value they may observe during the course of the computation.

\mypara{Secret-Sharing.} When we say a value is secret-shared in this work, we assume it is shared via a linear secret-sharing scheme (LSSS). In an LSSS, a value $x \in \F_p$ is said to be secret-shared among $r$ parties if each $P_i$, for $i \in [r]$, holds $\arith{x}_i$ such that $x = \sum_{i = 1}^r \arith{x}_i \mod p$. Given two secret-shared values $x$, $y$, we can compute a linear combination $z = a \cdot x + b \cdot y$, where $a, b$ are public constants, by carrying out the same operations on the respective shares; $\arith{z} = a \cdot \arith{x} + b \cdot \arith{y}$. 
Computing the product of two secret values, $\arith{z} = \arith{x} \cdot \arith{y}$, can also be done albeit more challenging. 
The guarantee secret-sharing schemes give is that, an adversary in possession of shares will not learn anything about the underlying secret, as long as the number of shares it possesses is below the reconstruction threshold. In our instantiation of the framework, we use the replicated secret-sharing scheme from~\citet{CCS:AFLNO16}.

\subsection{Security Model}
\label{sec:sec-model}

\mypara{Corruption.} There are three entities in our system: (1) clients, submitting model updates and validation data, (2) $\mpcnodes$, that work to verify the validity of the updates, and (3) the server that receives the aggregated update. 
We assume that up to half, or $< r/2$, of $\mpcnodes$ are corrupted (\ie honest majority). 
We assume $\cor$ out of $m$ clients are corrupt, and they could collude with $\mpcnodes$.
Since we work in a setting with a large number of clients, we assume that the fraction of corrupt clients is small (\eg $10\%$), as is standard in the literature.
Our framework supports two models of corruption. In the first one, we assume the clients are maliciously corrupt, meaning that both the updates and the validation data could be maliciously crafted, while the $\mpcnodes$ are passively corrupt (semi-honest). The second model is the stronger one to defend against, where we assume $\mpcnodes$ are also maliciously corrupt.

\mypara{Adversary's knowledge.}
We consider a gray-box scenario where the malicious parties know everything about the protocols and parameters in the framework except 1) the benign clients' model update and private data, and 2) the randomness in the local computation of benign clients. 
In addition, each malicious client $\calC_i$ has access to a clean dataset $\bfD_i$ drawn from the underlying distribution---the dataset stored at the client before being compromised.

\mypara{Adversary's goal.} 
First, the adversary attempts to reduce the utility of the honest parties in FL. In this paper, we mainly consider \emph{model poisoning}, in which the malicious clients attempt to maximize the loss of the global model at time $T$ by sending malformed model updates.
Let the first $m_c$ clients $\calC_1, \cdots, \calC_{m_c}$ be malicious without loss of generality. The adversarial goal is
\begin{equation}
\label{obj:adv}
    \max_{\bfu^{(t)}_i, i\in[m_c], t\in[T]}\ \ \  \calL(\bfw^{(T)}, \calD_{test}),
\end{equation}
where $\bfu^{(t)}_i$ is the malformed update submitted by $\calC_i, i\in [m_c]$ at time step $t$.
$\calD_{test}$ is the test distribution and $\calL(\bfw^{(T)}, \calD_{test})$ is the loss of the global model at the final round $T$ on the test distribution. 
Second, the malicious parties also want to undermine the privacy of honest parties by inferring as much information about the private data and/or the model updates of the honest clients.

\subsection{Our Goals}
\label{sec:goals}

\mypara{Robustness of the global model against adaptive attacks.}
In response to the adversarial goal in Equation~\ref{obj:adv}, we seek to preserve the integrity of the global model as specified in the following min-max objective,
\begin{align}
\label{obj:defense}
\min_{\calA}\max_{\bfu^{(t)}_i, i\in[m_c], t\in[T]}\ \ \  \calL(\bfw^{(T)}, \calD_{test}),
\end{align}
That is, we want to find an aggregation mechanism $\calA$ to minimize the test loss of the final global model $\bfw^{(T)}$ under the \textit{strongest} possible adversarial attack.
Specifically, assuming an adversary with knowledge of $\calA$ and control over local training data and local model updates on the malicious clients, one should consider the robustness against the strongest realizable attack devised considering the assumption. 
Accordingly, we introduce an aggregation protocol and report its evaluation result against its adaptive attack.

\mypara{Adaptability to Distribution Shift.} 
The underlying data
distribution of an FL task can change in common scenarios. New clients joining an FL protocol may have slightly different data distribution from existing clients, \eg in health care, new hospitals from a different geographic location with different patient demographics may want to join an existing federation. Clients already in the protocol may also see different data distribution over time, \eg a hospital's patient demographic may shift due to changes in the local population, disease outbreaks, or the introduction of new healthcare programs. 
Consider a global model converged with respect to an original client distribution $\calD_{old}$. 
Let $\calD_{new}$ represent the update distribution, incorporating both original and shifted client data.
Depending on how different $\calD_{new}$ is from $\calD_{old}$, the accuracy of the global model is bound to suffer for a period of time, or may not even be able to adapt to $\calD_{new}$. 
We address this aspect as \emph{adaptability} and report the global model accuracy on the test set sampled from the changing client data distribution over time. 
None of the prior works on robust secure aggregation baselines were designed or evaluated in this context, even though distribution shifts are ubiquitous in real-world FL deployments.
Particularly, the primary limitation of prior approaches stems from their reliance on public datasets for determining aggregation parameters~\cite{CCS:CGJv22}; in PPFL, maintaining a public dataset representative of the evolving data distribution is not straightforward. 
In contrast, our framework provides a way to take a reference of evolving client local data securely, hence naturally leading to adaptive aggregation results.

\mypara{Privacy of the client updates.} Our framework guarantees that even when up to half the aggregation committee and the clients are maliciously corrupted, nothing can be learned based on the \emph{intermediate values} observed during the computation, about the honest clients' updates beyond what is permissible by the algorithm by any other party in the protocol. 

\mypara{Correctness of secure aggregation.} Assuming the MPC protocol instantiated is secure with abort, we guarantee that if the protocol terminates without the adversary aborting, the final aggregate has been computed correctly.

\section{\ours: an Overview}
\label{sec:framework}

At a high level, our framework securely leverages client's private data to verify the integrity of model updates. It comprises two main components: 1) a secure and private cross-client check procedure that computes useful statistics for robustness, \eg the accuracy of a client's model on another client's data, to determine the weight of each model update, and 2) a secure aggregation protocol that calls the aforementioned cross-client check procedure and computes the global model for the next round.

We describe an MPC-computation friendly filtering mechanism based on a simple check score in~\secref{slvr-check}, and proceed to the MPC protocol that facilitate the secure computation in~\secref{slvr-protocol}.

\subsection{Check Procedure for Robustness}
\label{sec:slvr-check}
Our check procedure determines the weight of each client's model update in aggregation. As shown in~\figref{pcheck}, it consists of 1) a check committee selection for each client, 2) a check function/score for each client's update, and 3) a weight assignment to each update.  
At the start of training, the server $\server$ specifies a subroutine that determines the weights of client updates to the MPC nodes $\mpcnodes$. A client $\calC_i$ secret-share a random subset of their training data to $\mpcnodes$, which we refer to as \textit{validation data}, denoted as $\bfD^{\msf{val}}_i$.

\mypara{Check Committee.} Let $\cor$ denote the number of malicious clients. 
$\mpcnodes$ sample a set of $2 \cor + 1$ clients for each client $\calC_i$. This set of clients, denoted by $\calC_{check, i}$, is called the check committee for $\calC_i$. Validation data submitted by these clients will be used to check $\calC_i$'s model. The size $2\cor+1$ guarantees an honest-majority in $\calC_{check, i}$.

\mypara{Check Score.} Intuitively, a check score should reflect how much positive contribution a client's update make to the global objective. Let $\score{i}{j}$ be the score assigned to $\calC_i$.
A choice of check score can be defined as,
\begin{align}
\label{eqn:check_acc}
    &\score{i}{j} = \msf{Chk}_{acc}(\bfw_i, \bfD^{\msf{val}}_j) - \msf{Chk}_{acc}(\bfw^{(t - 1)}, \bfD^{\msf{val}}_j), \\ \nonumber
    &\msf{Chk}_{acc}(\bfw, \bfD) = \frac{1}{|\bfD|} \sum_{(\bfx_k, y_k) \in \bfD} \indicator{\hat{y}(\bfx_k; \bfw) = y_k}
\end{align}
where $\bfw^{(t - 1)}$ is the global model at the end of the previous round.
$\bfw^{(t)}_i = \bfw^{(t - 1)} + \bfu^{(t)}_i$ is $i$-th client's local model at round $t$ (we omit the superscript $t$ for the client model for brevity).
Let $\bff(\cdot; \bfw)$ be a neural network classifier parameterized by $\bfw$, and 
$\hat{y}(\bfx; \bfw) = \argmax_l \bff_l (\bfx; \bfw)$ is the predicted label given $\bfx$.

Essentially, $\score{i}{j}$ measures the increase of accuracy of $\calC_i$'s new model on $\bfD^{\msf{val}}_j$ over the current global model~\footnote{We note that $\score{}{}$ is not limited to the accuracy difference in Equation~\ref{eqn:check_acc}. See Appendix~\ref{app:softscore} for a `soft' alternative using confidence score, which corresponds to \ours(prob) in Sec~\ref{sec:experiment}.}.
Last, we condense the set of scores $\{\score{i}{j} \}_{j \in \calC_{check, i}}$ into a scalar $\score{i}{}$ as the check score for $\calC_i$'s model. We define $\score{i}{}$ to be the trimmed mean of $\{\score{i}{j} \}_{j \in \calC_{check, i}}$ to enhance numerical stability.

\mypara{Weights Assignment.}
We keep the top $k\%$ of the model updates in the final aggregation, where $k = 1-\cor/m$ is a public threshold.

\mypara{Synergy with Norm Bound.} Our check score and the norm bound are naturally complementary. During the early epochs, validation loss alone may not be the most stable indicator of malicious update; the standard norm bound constraint can effectively limit adversarial impact. As the model converges, the cross-client check becomes increasingly effective since the statistics used in the check are closely associated with the learning objective.

\subsection{Secure Aggregation with Check Results}
\label{sec:slvr-protocol}

\figref{pisecagg} shows the secure computation backbone of \ours. The most critical components are 1) a setup phase to establish the necessary security keys, 2) a secret-sharing ($\PShare$) and a reconstruction protocol ($\PRec$) protocol pair, and 3) MPC subroutines for operations in the check algorithms such as inference ($\FSecInf$), random selection ($\FRand$) and sorting ($\FSort$).
We elaborate our instantiation in Appendix~\ref{app:mpcdetails}.

\begin{protofig}{Protocol $\PSecAgg$}{Protocol for Robust Secure Aggregation}{pisecagg}
\label{proto:ours-acc}

\textbf{Parameters:} MPC nodes $\mpcnodes$, the clients $\calC_1, \ldots, \calC_m$. Number of corrupted clients, $\cor$. The threshold $k$ to select the top-$k$ updates, number of validation datapoints, $d$. \\

\textbf{Setup:} $\mpcnodes$ call $\FRand$ to receive shared pairwise keys, $\kappa_{ij}$, for $i, j \in |\mpcnodes|$ and a key, $\kappa$, shared between all the parties. \\

\textbf{Validation committees:} For each client $\calC_i$, $\mpcnodes$ use their common key, $\kappa$, to sample a sets of clients, $\calC_{check, i}$, where $|\calC_{check, i}| = 2 \cor + 1$. \\

\textbf{Model Updates:} At the start of every round, each client and $\mpcnodes$ run $\PShare$ on its local model, $\vec u_i$, and the validation data $\vec D_i$,  where $|\vec D^{\msf{val}}_i| = d$, to generate $\arith{\vec u_i}, \arith{\vec D^{\msf{val}}_i}$ with $\mpcnodes$. \\

\textbf{Robustness Check and Secure Aggregation:} 
\begin{enumerate}
    \item $\mpcnodes$ run $\PCheck$ with input $\vec w_i$, the global model from the previous round, $\vec w$, and $\vec D^{\msf{val}}_i$, for $i \in [m]$. 
    \item They receive $\arith{b_i}$, for $i \in [m]$, where $b_i = 1$ if $\calC_i$'s model has passed the check, and 0 otherwise.
    \item Compute $\arith{\vec w'_i} = \arith{\vec w_i} \cdot \arith{b_i}$, by running $\PMult$.
    \item Locally compute $\arith{\vec w_{\msf{aggr}}} = \Sigma_{i = 1}^m \arith{\vec w'_i}$.
    \item Reconstruct $\vec w_{\msf{aggr}}$ by running $\PRec$.
\end{enumerate}

\end{protofig}

\mypara{Security and Privacy Properties.} Security and privacy are straightforward to argue for our protocol, as the protocol is only composed of calls to the ideal functionalities. The parties do not communicate messages with each other outside of these calls, and only perform local computation. Our choice of the check function, $\msf{Chk}$, does not impact privacy as the result of the check is never revealed to the parties. The results are only used to obliviously compute the weight multiplier for each client's model update, and the final aggregated model is revealed to the server.
We formally state the following theorem, and give a proof sketch in \appref{proofs}.

\begin{protofig}{Protocol $\PCheck$}{Protocol for Robustness Check}{pcheck}
\label{proto:pcheck}

The protocol receives a set of client updates to check, $\vec u_i$, and validation datasets $\vec D^{\msf{val}}_i$, for $i \in [m]$, to check against. $\lambda$ is a public multiplier for the bound. \\

\textbf{Norm Check:} To compute the norm bound, $\mpcnodes$ do the following,
\begin{enumerate}
    \item Compute $\arith{\msf{norm}_i} = \sqrt{\arith{\vec u_i^2}}$ by running $\PMult, \PSqrt$. Call $\FSort$ on the vector $\arith{\msf{norm}} = \{ \arith{\msf{norm}_i} \}$ for $i \in [m]$. 
    \item Set $\msf{bound}$ = $\lambda * \msf{Median}(\{ \msf{norm}_i \})$, for $i \in [m]$. 
    \item Compute the vector $\{ \arith{b^{\msf{norm}}_i} \}$ via $\FComp$, for $i \in [m]$, where $b^{\msf{norm}}_i = 1$ if $\msf{norm}_i < \msf{bound}$, and 0 otherwise.
\end{enumerate}

\textbf{Accuracy Check:} For each client, $\calC_i$,
\begin{enumerate}
    \item Call $\FSecInf$ with $(\vec w_i, \vec D^{val}_j)$, and $(\vec w, \vec D^{val}_j)$, to receive $\arith{\Acc{i}{j}}$ and $\arith{\Acc{i}{\msf{t - 1}, j}}$ respectively, where $\vec w$ is the global model from the previous round, for $j \in \calC_{check, i}$.
    \item $\mpcnodes$ locally compute $\arith{\diff{i}{j}} = \arith{\Acc{i}{j}} - \arith{\Acc{i}{t - 1, j}}$.
    \item Call $\FSort$ on $\arith{\diff{i}{j}}$, for $j \in \calC_{check, i}$, to receive a sharing of the sorted scores.
    \item Let $\msf{trim \mhyphen diff}^j_i$, for $j \in \cor + 1$, denote the list each $P_k \in \mpcnodes$ obtains by locally trimming the first and last $\cor/2$ elements from the sorted scores. 
    \item Locally compute $\arith{\msf{scr}_i} = \Sigma_{j = 1}^{\cor + 1} \arith{\msf{trim \mhyphen diff}^j_i} / (\cor + 1)$. 
\end{enumerate}

\textbf{Select top-$k$:}
\begin{enumerate}
    \item Call $\FSort$ on the list $(i, \arith{\score{i}{}})$ to sort on the scores. Call $\FZeroOne$ with $k$ to receive $k$ shares of one and $m - k$ shares of zero. Set the top $k$ shares to be shares of one and the rest to be shares of zero, received from $\FZeroOne$.
    \item Call $\FSort$ on $(i, \arith{\score{i}{}})$, to sort on the indices, $i \in [m]$. 
    \item Compute $\arith{b_i} = \arith{\score{i}{}} \cdot \arith{b^{\msf{norm}}_i}$, for $i \in [m]$ via $\PMult$, and output the vector.
\end{enumerate}

\end{protofig}

\begin{theorem}\label{thm:secagg}
    Protocol $\PSecAgg$ securely realises the functionality $\FSecAgg$ in the presence of a malicious adversary that can statically corrupt up to $\cor < m/2$ parties in the protocol, in the $(\FSecInf, \FSort, \FRand)$-hybrid model.
\end{theorem}

\textbf{Output privacy:} We consider output privacy to be out of scope of this work. MPC formally only guarantees privacy of the inputs, and that the intermediate values do not leak any additional data than what is already allowed by the algorithm. Since all the parties in the system learn the aggregated global model at the end of every round of training, it is possible for an adversary to try to reverse engineer the model to infer data about the honest clients. A popular approach to mitigating this is using differential privacy, which involves sampling and adding noise to the model updates. We note that it is an interesting future direction to combine differential privacy with our framework.

\section{Adaptive Attacks}
\label{sec:adaptiveattack}

In order to capture the constantly evolving real-world threats, we evaluate our framework against an \emph{adaptive} adversary which knows the robust aggregation and the check algorithms. Its objective can be characterized as,
\begin{align}
\label{obj:adaptiveadv}
    & \max_{\substack{\bfw^{(t)}_i, \ \bfD^{\msf{val}}_i, \\ \forall i\in[\cor], t\in[T]}}\ \ \  \calL(\bfw^{(T)}, \calD_{test}),
\end{align}
\ie the adversary can manipulate both the \emph{client updates} and the \emph{validation dataset} at the corrupted clients to maximize the loss of the global model on clean test data. 

\mypara{Malicious Model Updates.} The generation of malicious model updates can be formulated as a constrained optimization problem: the attacker searches for the updates that diverge maximally from the learner's goal subject to the constraints imposed by the defense mechanism. We adopt the attack framework in~\citet{fang2020local} and instantiate attacks for all defense baselines (see Appendix~\ref{app:adaptive-attack}).

\mypara{Malicious Validation Data.}
A new attack surface for \ours\ is that the adversary can also manipulate malicious clients' validation data (\ie $\bfD^{\msf{val}}_i, \forall i \in [\cor]$). 
We evaluate the robustness of \ours\ against this new threat by considering the most potent attacker---each malicious client $\calC_{adv}$ can \emph{arbitrarily} manipulate the check score if the validation data is provided by $\calC_{adv}$. 
We propose a strong manipulation strategy called the extreme manipulation and prove its optimal attack strength as the following. 

\begin{definition}[Extreme Check Score Manipulation]
\label{def:optimalattack}
When $\calC_{j}$ is malicious, a check score manipulation is \emph{extreme} if
\begin{equation}
\label{eqn:extreme}
\score{i}{j} =
\begin{cases}
1 &\mbox{ if $\calC_i$ is malicious,}\\
-1 &\mbox{ otherwise.}
\end{cases}
\end{equation}
\end{definition}

\begin{theorem}(Optimality of the Extreme Manipulation.)
\label{thm:optimal}
    Let $\bfw^{(T)}_{adv}$ be a model that maximizes the objective in Formula~\ref{obj:adaptiveadv} and is obtainable by some combination of malicious model updates and validation data set manipulation. Then there must be a set of malicious model updates that lead to $\bfw^{(T)}_{adv}$ under the extreme check score manipulation in Definition~\ref{def:optimalattack}.
\end{theorem}

\begin{proof}
    First, we observe that the range of $\scr(i,j)$ in our protocol is from $-1$ to $1$.
    Consider a malicious $\calC_i$. When $C_j$ is malicious, $\scr(i, j)$ equals 1, which is the largest possible. When $C_j$ is benign, $\scr(i, j)$ is independent to malicious data manipulation. As a result, $\scr_i$ for a malicious client $\calC_i$ under the extreme manipulation is no less than under other manipulation. Similarly, $\scr_i$ for a benign client under the extreme manipulation will be no more than that under other manipulation. Therefore, the ranking of $\scr_i$ for a malicious client $\calC_i$ will never drop when switching from a non-extreme manipulation to extreme manipulation.

    Let $\bfw_{adv}$ be a malicious model update from $\calC_i$ in the sequence that leads to $\bfw^{(T)}_{adv}$ under some non-extreme validation data manipulation.
    \begin{enumerate}
        \item
        If $\bfw_{adv}$ is included in aggregation, then $\scr_i$ must already be in the top $k\%$ among all clients, $k$ being the threshold defined in the protocol. Since switching to the extreme manipulation will never reduce the ranking of a malicious $\scr_i$, $\bfw_{adv}$ will remain in aggregation. 
        \item
        If $\bfw_{adv}$ is not included in aggregation, then the attacker can send any $\bfw'_{adv}$ that will be excluded from aggregation as well.~\footnote{If such a $\bfw'_{adv}$ does not exist, i.e., all updates from the malicious client will be accepted, then we allow the attacker to refrain from submitting the updates. The extreme manipulation serves to check the robustness lower bound of our work. Therefore, we allow it more flexibility, i.e. more attacker strength.}
    \end{enumerate}
    Therefore, an adversary using extreme manipulation can always construct malicious model updates that have the same impact in each aggregation step, and thus obtain $\bfw^{(T)}_{adv}$.
\end{proof}

\mypara{Robustness Lower Bound.} We acknowledge that in practice, the adversary may not find an actual validation set that achieves the conditions in Equation~\ref{eqn:extreme}. Nevertheless, we equip the adaptive attack against our framework with the extreme check score manipulation in our empirical study in Section~\ref{sec:experiment}.
We implement the extreme manipulation by allowing the adversary to directly set the value of $\scr(i, j)$ when $\calC_j$ is a malicious client. By giving this additional power to the adversary, we can find an empirical \emph{lower} bound to the robustness of our work. 
We also note that it is still possible for an attacker to craft validation sets to boost scores for malicious clients while lowering those of benign ones in reality. 
They may flip labels to disrupt benign updates or embed backdoors that trigger high scores for malicious updates.
Thus, our empirical lower bound serves as a practical robustness indicator.

\section{Experiments}
\label{sec:experiment}

In this section, we employ experiments and show that~\ours~has (1) better robustness in the presence of malicious clients (Section~\ref{sec:empirical-robustness}; Table~\ref{tab:robustness}), (2) better adaptability to the change in the underlying distribution of client data (Section~\ref{sec: empirical-adaptability}; Figure~\ref{fig:adaptability}), and (3) reasonable computation and communication overhead that can be significantly improved via parallelism (Section~\ref{sec: mpc-benchmark}; Table~\ref{tab:mpc-benchmark}).

\subsection{Setup}
\label{sec:exp-setup}

We briefly describe the experimental setup following the conventions in the FL literature~\cite{fang2020local,CCS:CGJv22}.
See Appendix~\ref{app:setup} for full descriptions.

\mypara{Models and Datasets.} We consider the following model architectures and datasets: LeNet-5 for MNIST and Fashion-MNIST (FMNIST), ResNet-18 for SVHN and CIFAR-10.
We follow the given train vs. test set split and further split each train set, balanced across classes. We follow the setup in~\cite{CCS:CGJv22} and reserve random 10k training samples as the public validation set ($\calD_{pubval}$) for aggregation baselines that leverage public validation datasets. The remaining training samples are partitioned across clients by a Dirichlet distribution $\text{Dir}_K(\alpha)$ with $\alpha = 0.5$ to emulate data heterogeneity in realistic FL scenarios~\cite{cao2020fltrust}.
We assume there are 100 clients where $10\%$ or $20\%$ of them are malicious clients.

\mypara{Defense baselines.} We compare the performance of \ours with the following aggregation methods that are most commonly considered in Byzantine-robust secure FL literature: 
\begin{enumerate}
    \item \textit{Norm Bound (adaptive)}~\cite{rofl, elsa} accepts a client update if the update is bounded by a $\tau$: $\msf{Chk}(\bfu_i) = \indicator{||\bfu_i^{(t)}||_2 < \tau}$.
    At the start of the $t$th communication round, the server computes the threshold as $\tau = \lambda \times \text{median}(\bfu_1^{(r-1)}, \cdots, \bfu_m^{(t-1)})$ adaptively, and then broadcast it to all clients. Hence, each client can submit the norm-bound check result based on the threshold value.

    \item \textit{Norm Bound ($\calD_{pubval}$ or public data)}~\cite{CCS:CGJv22} is identical to the above Norm Bound (adaptive) except that the threshold is computed by referring to the public validation dataset: $\tau = ||\bfu_{pubval}||_2$ where $\bfu_{pubval}$ is the previous round's global model update computed with the public validation data $\calD_{pubval}$. $\tau$ can be computed by anyone participating in the communication.

    \item \textit{Norm Ball}~\cite{steinhardt2017normball,CCS:CGJv22} accepts a client update if the update is within a spherical radius from the reference %
    update computed using the public validation dataset: $\msf{Chk}(\bfu_i, \calD_{pubval}) = \indicator{||\bfu_i^{(t)} - \bfu_{pubval}||_2 < \tau}$ where $\tau = \lambda \times ||\bfu_{pubval}||_2$. $\bfu_{pubval}$ and the value for $\tau$ is available to anyone participating in the communication.

    \item \textit{Cosine Similarity}~\cite{cao2020fltrust,CCS:CGJv22} checks the cosine similarity between each client update and the global model update from the previous round $\bfu^{(t-1)}$: $\msf{Chk}(\bfu_i, \calD_{pubval}) = \indicator{\cos(\bfu_i^{(t)}, \bfu^{(t-1)}) < \tau}$ where $\cos(u, v) = \frac{<u, v>}{||u||_2 ||v||_2}$. The threshold is computed as $\tau = \lambda \times \cos(\bfu_{pubval}, \bfu^{(t-1)})$, and again, can be computed by anyone in the communication.
\end{enumerate}

$\lambda$ is a constant where a larger $\lambda$ allows for more gradient updates, leading to faster convergence but reduced robustness.

\mypara{Attack baselines.}
We evaluate each aggregation method against the following attacks:
\begin{enumerate}
    \item \textit{Additive Noise}~\cite{li2019additive} adds Gaussian noise to a malicious client update.
    \item \textit{Sign Flipping}~\cite{damaskinos2018signflipping} flips the sign of a client update: $\bfu'_i = - \bfu_i$, for a client $\calC_i$.
    \item \textit{Label Flipping}~\cite{fang2020local} is a data poisoning attack that flips the label of each training instance. Specifically, it flips a label $l \in \{0, 1, \dots, L-1\}$ into $L-l-1$ where $L$ is the number of classes.
    \item \textit{Adaptive Attack} is considered the strongest attack adaptive to each aggregation as described in Section~\ref{sec:adaptiveattack}.
\end{enumerate}

\definecolor{lightgreen}{gray}{0.9}
\definecolor{lightblue}{RGB}{217, 234, 243}
\definecolor{lightgreen}{RGB}{226, 240, 217}
\definecolor{lightyellow}{RGB}{255, 249, 219}
\DeclareRobustCommand{\hlgreen}[1]{\colorbox{lightgreen}{#1}}
\DeclareRobustCommand{\hlblue}[1]{\colorbox{lightblue}{#1}}

\begin{table}[t]
\tiny
\caption{\textbf{Test accuracy against adversarial attacks by 10\% malicious clients after 500 rounds.} Robustness is determined by the lowest accuracy (\textbf{boldfaced}) across all attacks (\ie the smallest value across the columns of each table). \ours\ maintains its robustness even against the strongest adaptive attacks (``Adaptive'').}
    \subfloat[LeNet-5 model, MNIST dataset]{
    \begin{adjustbox}{width=.7\textwidth,center}
    \begin{tabular}{c|cccc}
    \cline{1-5}
        & Additive Noise & Labeflip & Signflip & Adaptive  \\ \hline \hline
        Norm Bound (adaptive) & 0.970 & 0.983 & 0.981 & \bf 0.883 \\
        Norm Bound (public data) & 0.976 & 0.976 & 0.976 & \bf 0.454 \\
        Norm Ball & 0.968 & 0.985 & 0.966 & \bf 0.965 \\
        Cosine Similarity & 0.981 & 0.964 & 0.947 & \bf 0.863 \\
        \ours (acc) & 0.974 & 0.977 & 0.975 & \cellcolor{lightgreen}\bf 0.974\\
        \ours (prob) & 0.973 & 0.977 & 0.974 & \cellcolor{lightgreen}\bf 0.970 \\
        \cline{1-5}
    \end{tabular}
    \end{adjustbox}
    }
    \vspace{2em}

    \subfloat[LeNet-5 model, FMNIST dataset]{
    \begin{adjustbox}{width=.7\textwidth,center}
    \begin{tabular}{c|cccc}
    \cline{1-5}
        & Additive Noise & Labelflip & Signflip & Adaptive \\ \hline \hline
        Norm Bound (adaptive) & 0.872 & 0.873 & 0.867 & \bf 0.709\\
        Norm Bound (public data) & 0.865& 0.870& 0.863 & \bf 0.583\\
        Norm Ball & 0.870 & 0.872& 0.866 & \bf 0.798\\
        Cosine Similarity & 0.863& 0.862& 0.862 & \bf 0.714\\
        \ours (acc) & 0.870& 0.873& 0.865 & \cellcolor{lightgreen}\bf 0.830\\
        \ours (prob) & 0.870& 0.871& 0.864 & \cellcolor{lightgreen}\bf 0.832\\
        \cline{1-5}
    \end{tabular}
    \end{adjustbox}
    }
    \vspace{2em}

    \subfloat[ResNet-20 model, SVHN dataset]{
    \begin{adjustbox}{width=.7\textwidth,center}
    \begin{tabular}{c|cccc}
    \cline{1-5}
        & Additive Noise & Labeflip & Signflip & Adaptive \\ \hline \hline
        Norm Bound (adaptive) & 0.915 & 0.932 & 0.927 & \bf 0.198 \\
        Norm Bound (public data) & 0.921 & 0.923 & 0.919 & \bf 0.189\\
        Norm Ball & 0.913 & 0.917 & 0.901 & \bf 0.308 \\
        Cosine Similarity & 0.844 & 0.919 & 0.886 & \bf 0.431\\
        \ours (acc) & 0.868 & 0.925  & 0.922 & \cellcolor{lightgreen}\bf 0.836\\
        \ours (prob) & 0.874 & 0.927 & 0.920 & \cellcolor{lightgreen}\bf 0.855 \\
        \cline{1-5}
    \end{tabular}
    \end{adjustbox}
    }
    \vspace{2em}

    \subfloat[ResNet-20 model, CIFAR-10 dataset]{
    \begin{adjustbox}{width=.7\textwidth,center}
    \begin{tabular}{c|cccc}
    \cline{1-5}
        & Additive Noise & Labeflip & Signflip & Adaptive \\ \hline \hline
        Norm Bound (adaptive) & 0.778 & 0.791& 0.789 & \bf 0.218 \\
        Norm Bound (public data) & 0.784& 0.786& 0.799 & \bf 0.122\\
        Norm Ball & 0.771& 0.798& 0.779 & \bf 0.291\\
        Cosine Similarity & 0.770 & 0.781 & 0.775 & \bf 0.221\\
        \ours (acc) & 0.772& 0.781& 0.776 & \cellcolor{lightgreen}\bf 0.733\\
        \ours (prob) & 0.775& 0.784& 0.772 & \cellcolor{lightgreen}\bf 0.735\\
        \cline{1-5}
    \end{tabular}
    \label{tab:robustness-cifar10}
    \end{adjustbox}
    }
\label{tab:robustness}
\end{table}

\subsection{Robustness against Poisoning Attacks}
\label{sec:empirical-robustness}

\mypara{\ours~is more Byzantine-robust.}
Table~\ref{tab:robustness} summarizes the robust accuracy of each aggregation method under attacks.
According to our goal in Equation~\ref{obj:defense},
All aggregation methods demonstrate comparable robust accuracies against most cases of non-adaptive, common attack baselines (\ie Additive Noise, Labelflip, and Signflip), varying by only a small margin of $2 \mhyphen 3\%$ accuracy in each column.
However, all other defense baselines suffer from significant accuracy drops under adaptive attacks, and the robustness gap compared to \ours~becomes more pronounced with more complex tasks. For example, while the accuracy gap on MNIST is relatively modest, it widens to approximately 50\% on CIFAR-10.
In other words, unlike other aggregation baselines that are more susceptible to certain attacks (\eg Non-adaptive vs. Adaptive) or fail on certain data sets (\eg Norm Ball on CIFAR-10), \ours~shows consistent robustness against all attacks on all data sets.
Such consistency makes it a more desirable choice in real-world uncertainty. An adversary may employ the strongest possible attack, which is the adaptive attack in most cases, it is crucial to provide robustness even under those attacks; specifically, as in Figure~\ref{fig:adaptiveattack}, in the worst case, the adversary can only pull down the performance of FL with \ours~to ~97\% on MNIST, ~87\% on FMNIST, ~83\% on SVHN, and ~72\% on CIFAR-10 after 1000 rounds.

\begin{figure*}[ht]
\captionsetup[subfigure]{justification=centering}
    \centering
    \begin{subfigure}{0.24\textwidth}
        \includegraphics[width=\textwidth]{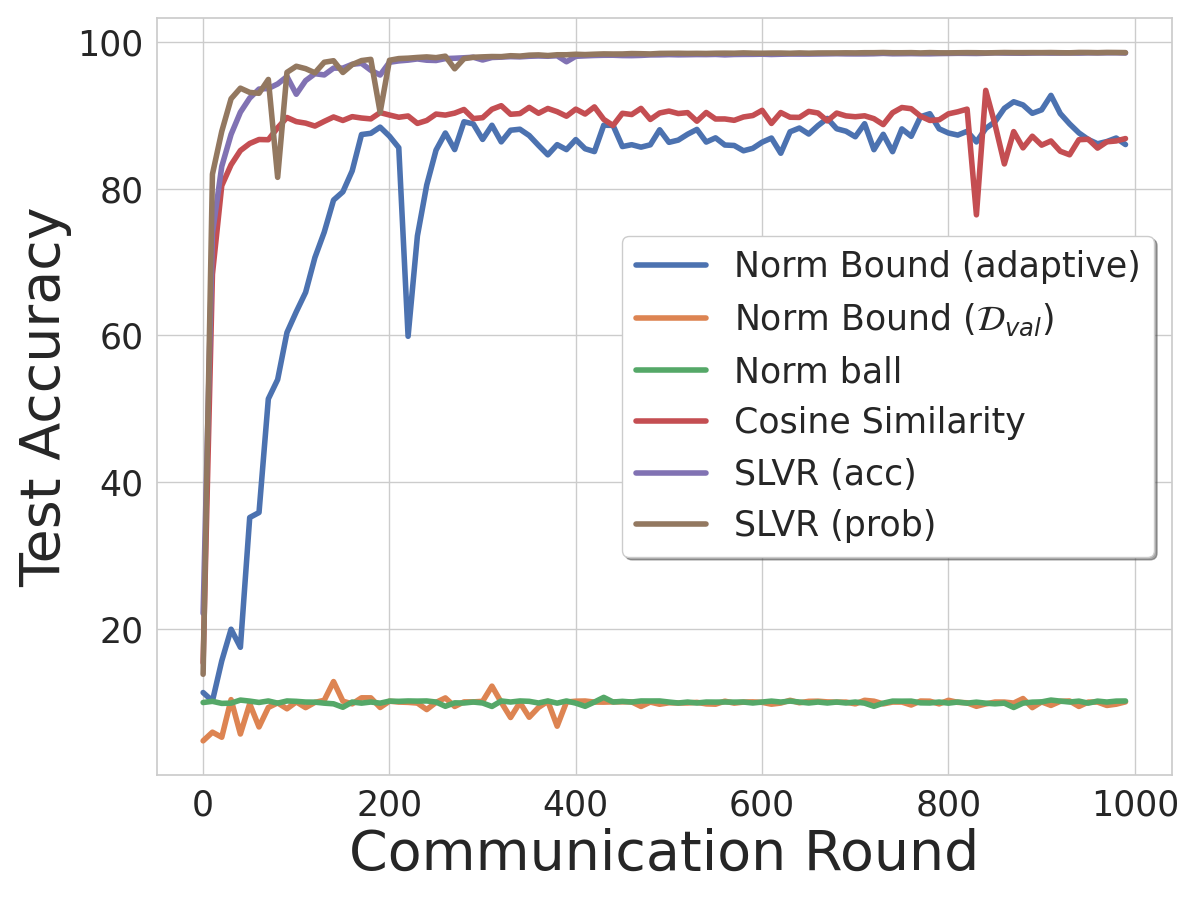}
        \caption{LeNet-5, MNIST}
        \label{subfig:adaptive-mnist}
    \end{subfigure}
    \hspace{.2mm}
    \begin{subfigure}{0.24\textwidth}
        \includegraphics[width=\textwidth]{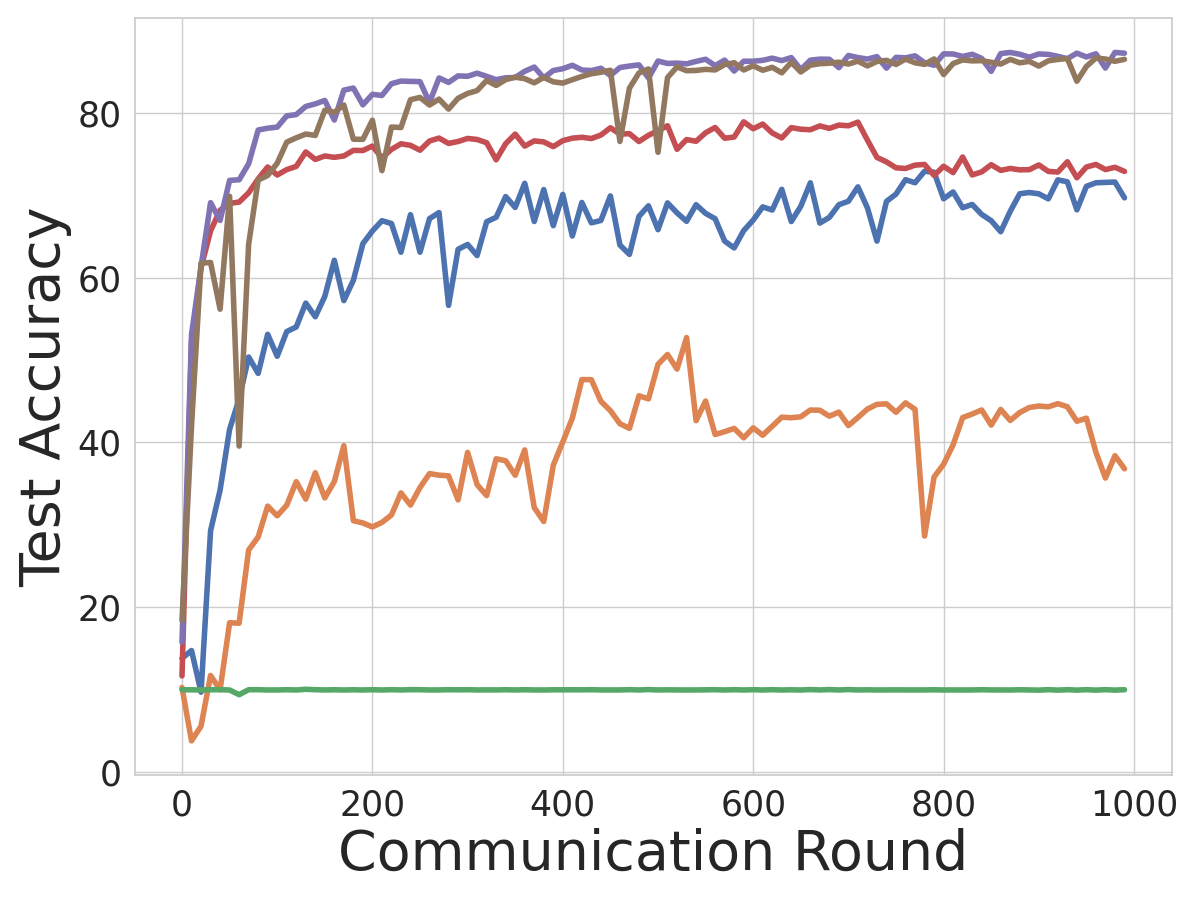}
        \caption{LeNet-5, FMNIST}
        \label{subfig:adaptive-fmnist}
    \end{subfigure}
    \hspace{.2mm}
    \begin{subfigure}{0.24\textwidth}
        \includegraphics[width=\textwidth]{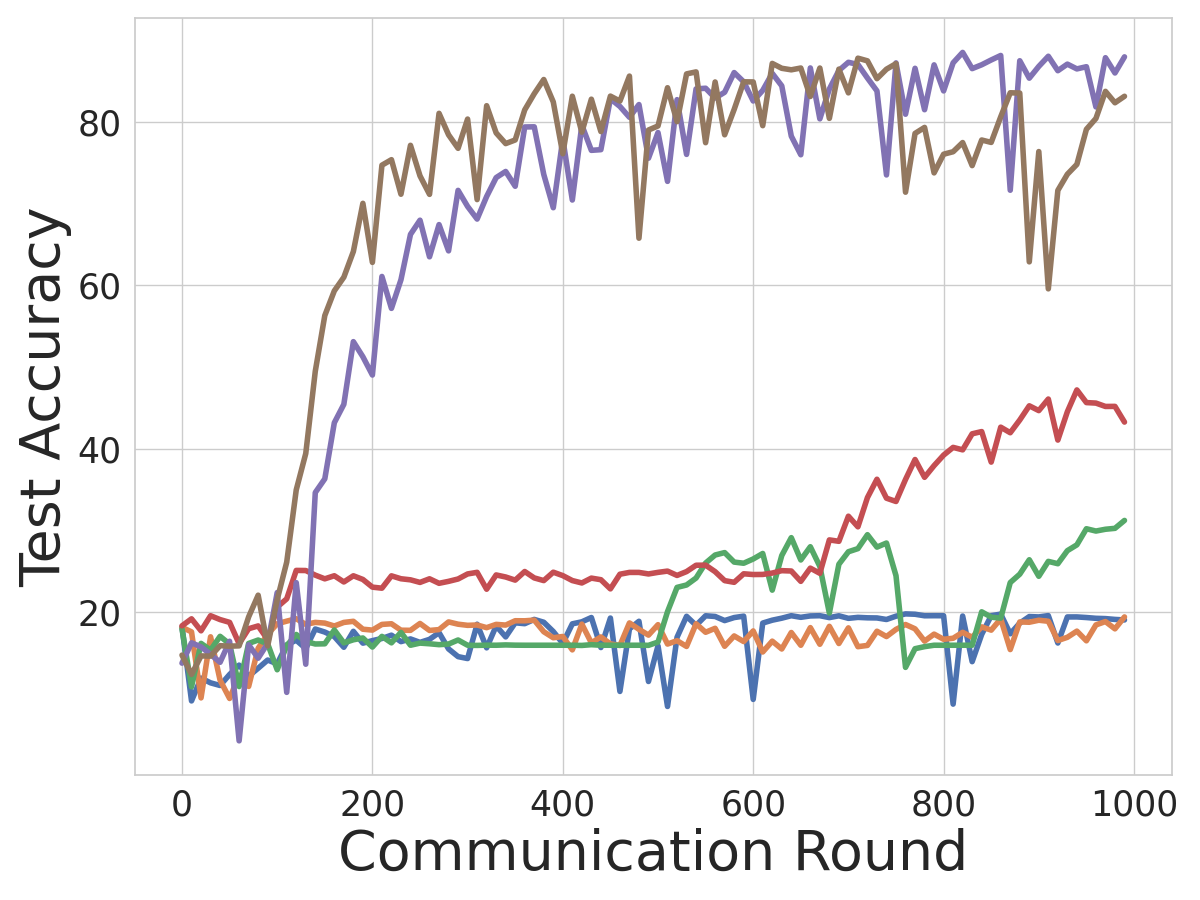}
        \caption{ResNet-20, SVHN}
        \label{subfig:adaptive-svhn}
    \end{subfigure}
    \hspace{.2mm}
    \begin{subfigure}{0.24\textwidth}
        \includegraphics[width=\textwidth]{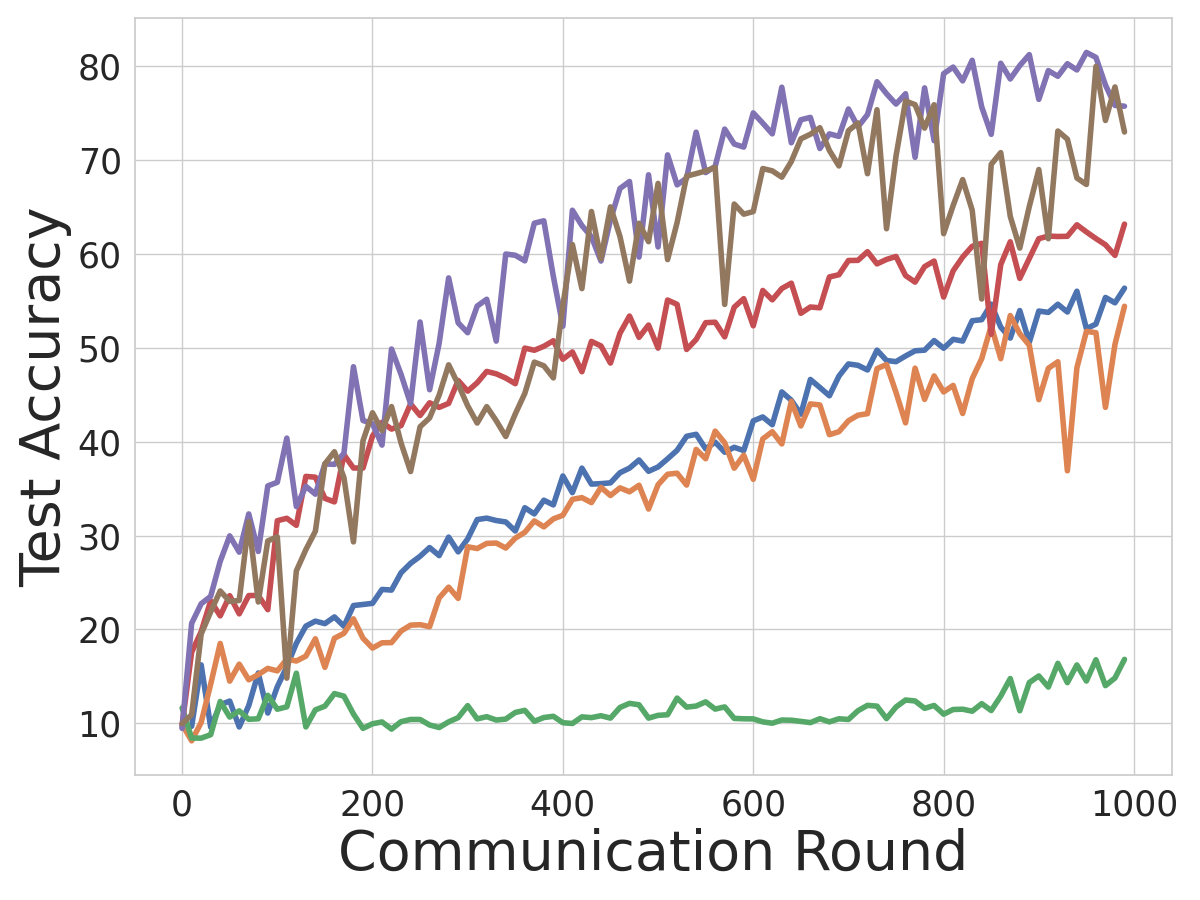}
        \caption{ResNet-20, CIFAR-10}
        \label{subfig:adaptive-cifar10}
    \end{subfigure}
    \caption{\textbf{Robustness against Adaptive Attack with 20\% malicious clients.} Both \ours~(acc) and \ours~(prob) demonstrate competitive convergence performance against adaptive attacks. In contrast, other baselines that rely on static public validation datasets (\eg Norm Bound ($\calD_{val}$), Norm Ball, Cosine Similarity) or those that are adaptive but rely on simple validation checks (\eg Norm Bound (adaptive)) become vulnerable.
    }
    \label{fig:adaptiveattack}
\end{figure*}

\begin{figure*}[t]
\captionsetup[subfigure]{justification=centering}
    \centering
    \begin{subfigure}{\textwidth}
        \centering
        \begin{subfigure}{0.49\textwidth}
            \centering
            \includegraphics[width=0.48\textwidth]{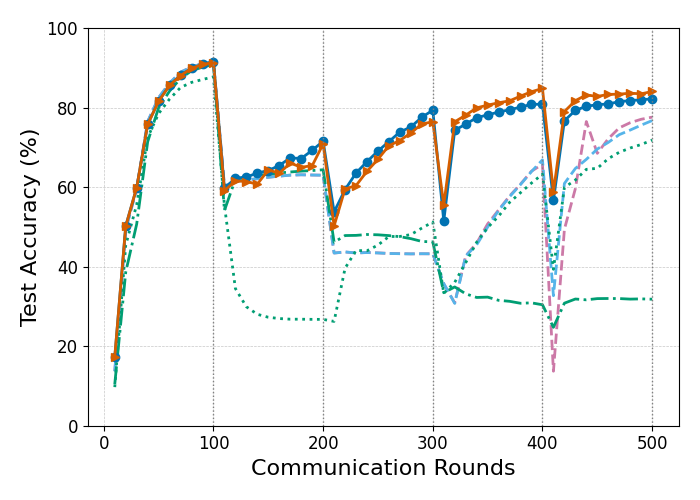}
            \includegraphics[width=0.48\textwidth]{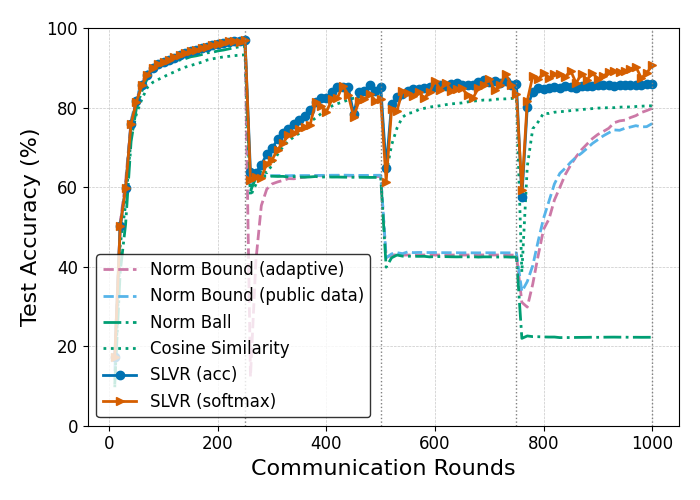}
            \caption{Scenario 1: evolving clients}
            \label{subfig:scenario-evolving}
        \end{subfigure}
        \hfill
        \begin{subfigure}{0.49\textwidth}
            \centering
            \includegraphics[width=0.48\textwidth]{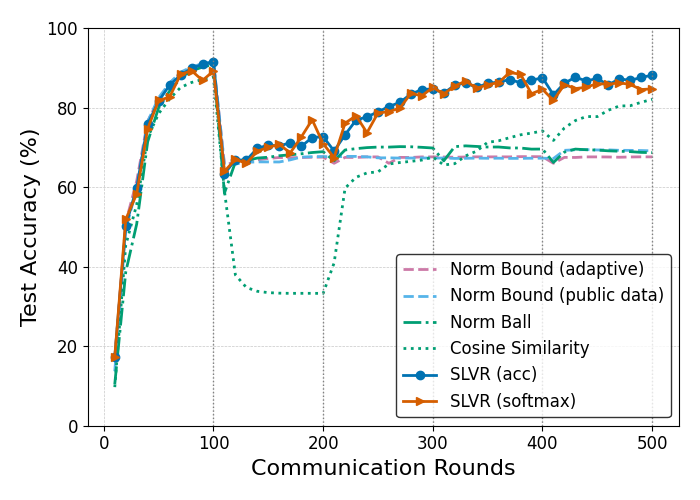}
            \includegraphics[width=0.48\textwidth]{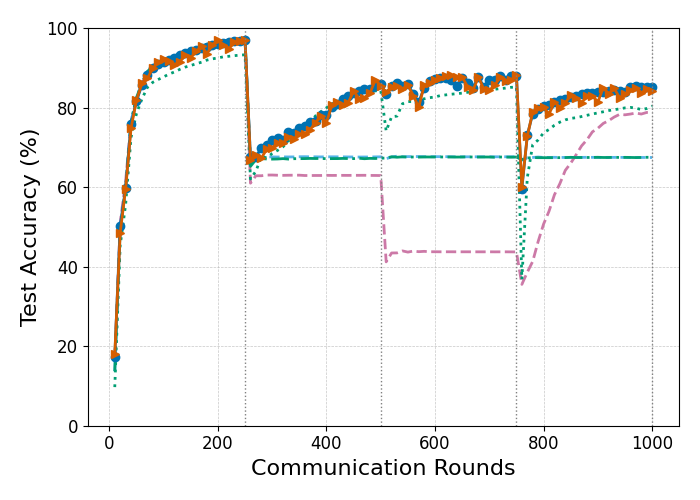}
            \caption{Scenario 2: new clients}
            \label{subfig:scenario-new}
        \end{subfigure}
    \end{subfigure}
    
    \caption{\textbf{Adaptability across the two distribution shift scenarios.} We start with 100 clients working on MNIST data. Every 100 (or 250) rounds (marked by dotted lines), either 20 random clients transition to SVHN data (\ref{subfig:scenario-evolving}), or an additional 20 clients with SVHN data join the communication (\ref{subfig:scenario-new}). Our model consistently adapts and improves accuracy throughout the communication rounds, outperforming other aggregation protocols that struggle to adjust to distribution shifts.}
    \label{fig:adaptability}
\end{figure*}

\subsection{Adaptability to Distribution Shifts}
\label{sec: empirical-adaptability}

In this section, we examine the adaptability of \ours~along with other baseline aggregation protocols in the face of distribution shifts.
We take into account two practical scenarios that may frequently occur in real-world FL settings:

\begin{enumerate}
    \item \textbf{Scenario 1: evolving clients} (Figure~\ref{subfig:scenario-evolving}).
    Some clients involved in the communication send different data over time to the central server. Consider a cross-silo FL system for predicting patient outcomes in hospitals. Over time, a hospital's patient demographic may shift due to various factors such as changes in the local population, disease outbreaks, or the introduction of new healthcare programs. This could lead to an evolution in the type of patient data the hospital, or a subset of hospitals send to the central server, reflecting these changes.

    \item \textbf{Scenario 2: new clients} (Figure~\ref{subfig:scenario-new}). Along with the existing clients, new clients join the communication, but their data differs slightly from the existing clients. For instance, in the above FL example for a healthcare application, the existing clients could be hospitals from a different geographic location, or specialized hospitals and clinics that want to join the system to solve the same task. The new clients' data coming from such sources may differ due to their specialized nature of care, different patient demographics, or unique healthcare practices.
\end{enumerate}

We simulate these scenarios using MNIST for existing clients and SVHN for evolving or new clients, with LeNet-5 as the model. The server initially communicates with 100 MNIST clients and uses a static MNIST validation dataset. Distribution shifts occur every 100 or 250 rounds, where clients transition to SVHN (Scenario 1) or new SVHN clients join (Scenario 2). We measure global model accuracy on MNIST initially and on the combined MNIST-SVHN test set thereafter. Further details are provided in Appendix~\ref{app:setup}.
We aim to evaluate how cross-client checks contribute to adaptability under changing data distributions. Therefore, we disable the norm-bound check in \ours~and rely exclusively on cross-client checks to address the occurrence of distribution shifts.

\mypara{\ours~can adapt to changing client data distributions.}
Figure~\ref{fig:adaptability} shows the convergence of the global model against the number of communication rounds.
We observe that baseline aggregation protocols relying on a static public validation dataset suffer significant performance degradation under distribution shifts.
Specifically, Norm Bound ($\calD_{pubval}$) and Norm Ball reject every update from SVHN clients because the magnitude of updates computed on SVHN far exceeds the threshold determined on the static MNIST validation set.
The same observation applies to Norm Bound (adaptive) that uses the median of gradient updates to adaptively set the threshold since the gradient updates computed with new SVHN data are likely to fall above the median. One could relax the threshold by multiplying larger constant values to the computed median of gradients, but again, this highlights that appropriate threshold parameter setting is crucial to ensure robustness and adaptability under changing distributions, which requires non-trivial efforts, especially in secure FL settings.
On the contrary, \ours~is capable of dynamically incorporating the new shifted local data into validation. Consequently, its choice of updates balances the performance on both MNIST and SVHN, which leads to much stronger adaptability without any parameter tuning.

\subsection{MPC Benchmarks}
\label{sec: mpc-benchmark}
We benchmark the computation and communication overhead of our MPC implementation. The framework is instantiated using building blocks from MP-SPDZ~\cite{CCS:Keller20}. Benchmarks were conducted on a MacBook Pro with an Apple Silicon M4 Pro processor and 48 GB RAM. We simulated the network setup locally, and ran all the parties on the same machine, with 4 threads. For the LAN case, we simulated a network with 1 Gbps bandwidth, and 1 ms latency, and for WAN, we used 200 Mbps and 20 ms respectively.

We benchmark the three components of our MPC protocol -- the norm bound check, the cross-client check, and selecting the top-$k$ model updates separately. In reality, the normbound and cross-client check can be run in parallel, with top-$k$ run on the outputs of these checks. All the components are instantiated using a semi-honest 3PC protocol from MP-SPDZ, specifically, replicated-ring-party.x which implements the replicated secret-sharing based protocol from \cite{CCS:AFLNO16}. The protocols are run with a 128-bit ring, although the larger ring size compared to the typical 64-bit ring is only to the size of the norm. It is possible to operate over a 64-bit ring, and switch to a 128-bit ring only for the norm, although we do not implement this optimization. 

\begin{table}[b!]
\centering
\caption{\textbf{MPC benchmarks.} With $\cor = 10\%$ corruptions, $2\cor + 1$ validation pairs, validation dataset size of 10.}

\begin{adjustbox}{width=\columnwidth,center}
\begin{tabular}{cc|cc|cc|cc|cc}
\toprule
     & & \multicolumn{2}{c|}{Normbound} & \multicolumn{2}{c|}{Cross-client check (acc)} & \multicolumn{2}{c|}{Cross-client check (prob)} & \multicolumn{2}{c}{Top-k filtering} \\ \cline{3-10}
     & & LAN & WAN & LAN & WAN & LAN & WAN & LAN & WAN \\
     \hline \hline
    \multirow{2}{*}{LeNet-5, 10 clients} & time (s) & 1.22 & 6.73 & 13.14 & 109.62 & 7.10 & 60.67 & 2.64 & 23.46 \\
    & data (MB) & \multicolumn{2}{c|}{26.40} & \multicolumn{2}{c|}{42.32} & \multicolumn{2}{c|}{21.51} & \multicolumn{2}{c}{0.25} \\ \cline{1-10}
    \multirow{2}{*}{LeNet-5, 20 clients} & time (s) & 1.78 & 8.92 & 20.20 & 174.90 & 10.89 & 93.39 & 2.65 & 23.53 \\
    & data (MB) & \multicolumn{2}{c|}{84.10} & \multicolumn{2}{c|}{70.08} & \multicolumn{2}{c|}{35.40} & \multicolumn{2}{c}{0.73} \\ \cline{1-10}
    \multirow{2}{*}{ResNet-18, 10 clients} & time (s) & 207.26 & 418.55 & 261.78 & 2084.90 & 131.30 & 1056.52 & 2.64 & 23.46 \\
    & data (MB) & \multicolumn{2}{c|}{4791.75} & \multicolumn{2}{c|}{2010.03} & \multicolumn{2}{c|}{1098.59} & \multicolumn{2}{c}{0.25} \\ \cline{1-10}
    \multirow{2}{*}{ResNet-18, 20 clients} & time (s) & 370.49 & 818.43 & 429.15 & 3471.01 & 217.45 & 1746.22 & 2.65 & 23.53 \\
    & data (MB) & \multicolumn{2}{c|}{9583.53} & \multicolumn{2}{c|}{2150.21} & \multicolumn{2}{c|}{1137.50} & \multicolumn{2}{c}{0.73} \\
   
\bottomrule
\end{tabular}
\end{adjustbox}

\label{tab:mpc-benchmark}
\end{table}

We report the MPC benchmarks in Table~\ref{tab:mpc-benchmark}. Compared to previous FL schemes, \ours\ unsurprisingly incurs increasing overhead as it supports private collection and computation of richer information via cross-client check. Nevertheless, the checks can be completed in reasonable time. Moreover, a large portion of the computations, \eg $\score{i}{j}$ for different $(i,j)$ pairs, are highly parallelizable, suggesting substantial speed-ups with more powerful infrastructure than our lab computer.

\section{Conclusions and Future Work}
To the best of our knowledge, our work is the first to take advantage of clients' private data via secure MPC for FL. We demonstrate massive potential of this class of checks at enhancing FL robustness and adaptability to distribution shifts, while preserving privacy. Moreover, the capability of securely computing statistics over multiple data sources has its use beyond robust FL. It can help more applications such as data appraisal, personalized FL and etc. 

We do note that our work has significant room for improvement especially in scalability. Currently, scaling our protocol to more clients and larger models or datasets remains computationally intensive. In addition to developing more efficient secure inference techniques, especially for larger models, we hope that more robustness checks based on statistics computed with private client data will be proposed. Integrating our work with differential privacy is also an interesting direction, to provide output privacy. We welcome the community to join our effort.

\bibliographystyle{icml2025}
\bibliography{shrunk, references}

\newpage
\appendix
\onecolumn
\section*{\centering \textbf{Appendices}}

\begin{table}[h]
\begin{adjustbox}{width=\textwidth,center}
\centering
\begin{tabular}{c|l}
\toprule
Symbol  & Description \\ \hline\hline
$m$ & number of all clients \\
$m_c$ & number of malicious clients \\
$t$ & communication round \\
$T$ & total number of communication rounds between the clients and the server \\
$\calS$ & server \\
$\calC_i$ & $i$th client \\
$\bfD_i$ & $i$th client's local training dataset \\
$N_i$ & number of $i$th client's local training data, $N_i = |\bfD_i|$ \\
$\bfu_i^{(t)}$ & $i$-th client's local update \\
$\bfw_i^{(t)}$ & $i$-th client's local model after local training with $\bfD_i$ at $t$ \\
$\bff_{\bfw}(\cdot)$ & a neural network classifier parameterized by $\bfw$ \\
$\mpcnodes$ & a set of MPC nodes \\
$r$ & number of parties in $\mpcnodes$ \\
$\calC_{check, i}$ & check committee of size $2\cor + 1$ for the $i$th client \\
$\bfD^{\msf{val}}_i$ & $i$-th client's validation data (a random subset of $\bfD_i$ that is secret-shared to the MPC nodes) \\
$\score{i}{j}$ & validation score for $i$-th client computed with $j$-th client's validation data\\
$\score{i}{}$ & aggregated validation score for $i$-th client\\
$\msf{Chk}_{acc}$ & check function based on accuracy difference \\
$\msf{Chk}_{prob}$ & check function based on softmax probability outputs \\
$\FSecInf, \FSort, \FMult, \FRand$ & Functionality for secure inference, oblivious sort, multiplication, sampling common random values \\
$\PShare, \PRec$ & Protocol for secret-sharing, reconstruction \\
$k$ & parameter in trimmed mean, $1 - \cor/m$ \\
$\alpha$ & Dirichlet distribution parameter that controls data heterogeneity across clients \\
$\calD_{pubval}$ & public validation data \\
$L$ & number of classes \\

\bottomrule
\end{tabular}
\end{adjustbox}
\caption{Symbols and notations used in the paper. }
\label{tab:notations}
\end{table}

\section{Details of $\PSecAgg$ and $\PCheck$}
\label{app:mpcdetails}
We found the best performance of our framework when we used both a norm bound and the accuracy function as ways to achieve robustness, as shown in \secref{experiment}. 
Both the checks can be instantiated from off-the-shelf MPC protocols, owing to their use of standard building blocks. We combine both the checks by carefully using the subprotocols required. We use subprotocols for secure inference, multiplication, oblivious sorting, while assuming a shared key setup. There has been a lot of interest in proposing efficient instantiations of sorting and secure inference in the recent years~\cite{CCS:AFOPRT21,ruffle,tetrad,CCS:MohRin18,USENIX:DalEscKel21,EPRINT:KelSun22}, offering protocols at various levels of security. The only restriction we have in terms of using these different frameworks is that, we require the secret-sharing to have an access structure such that the honest parties can reconstruct the secret. The variants of replicated secret-sharing \cite{EPRINT:FLNW16} that have been used in the recent PPML works, including the ones mentioned above, satisfy these criteria. We refer the readers to these, along with MP-SPDZ~\cite{EPRINT:Keller20} for detailed descriptions of the input sharing, $\PShare$, and reconstruction protocols, $\PRec$. For completeness, we provide the functionality for using the maximum softmax value in \appref{proofs}, but we do not report numbers from experiments with this function.

As mentioned earlier, our secure aggregation protocol assumes there exists a set of parties that can run MPC, denoted by $\mpcnodes$. For ease of description, one could assume that $\mpcnodes$ constitutes the servers that want to train the FL model together, as is the case in multi-server FL. However, it is possible to form an aggregation committee by sampling from the clients, as shown in works such as \cite{SP:MWAPR23,CCS:BBGLR20}. We leave it as a potential future direction to integrate our check into a protocol with a changing set of parties constituting $\mpcnodes$. 
We assume that the nodes run a Setup phase, which establishes shared PRF keys between every pair of parties, as well as a common PRF key between all of them. These keys can be used to generate common random values by keeping track of a counter, and incrementing it every time the function is called. In every round, $\mpcnodes$ use the keys to generate a set of $2 \cor + 1$ random IDs for each client, $\calC_i$, $i \in \{1, m\}$. Each client's model will be checked against validation data from the set of random IDs assigned to it.

In addition to the features we strive to achieve, such as not relying on public datasets, we also placed an emphasis on making the MPC implementation-friendly. Since secure inference is often part of the system when training an FL model, reusing that building block to also perform a robustness check makes the system easier to deploy. We describe high level details of the protocol below.

Each round $t$, begins with each client $\calC_i$, for $i \in [1, m]$, secret-sharing its model update to $\mpcnodes$ using $\PShare$. $\mpcnodes$ call $\FSecInf$ (\figref{fsecinf}) with the client's model, the global model from the previous round, and each of the validation datasets in $\calC_{check, i}$. As mentioned before, there are $2 \cor + 1$ client datasets against which the model is checked, where $\cor$ is the total number of corrupted clients. 

$\mpcnodes$ receive $\arith{\Acc{i}{\msf{new}, j}}$, for $j \in \calC_{check, i}$ and $\arith{\Acc{i}{\msf{prev}, j}}$, which denote the accuracy of the client $\calC_i$'s model and the global model of the previous round, against the validation client's data respectively. The difference between the two accuracies is denoted by $\diff{i}{j}$. These values are then sorted by calling $\FSort$ (\figref{fsort}), which sorts the array obliviously. $\mpcnodes$ then compute the trimmed mean of this vector. Since it is known to all parties that they need to trim the top and bottom $m_c/2$ elements from this array, they can do the computation locally and compute the mean of the remaining ones. The result of the trimmed mean is denoted by $\score{i}{}$ for each client $\calC_i$.

$\mpcnodes$ compute the score for each client, to get a secret-sharing of the list of scores, $\{ \score{i}{} \}_{i \in m}$. In order to eliminate the lowest-performing updates, the parties sort the scores and select the top $k \%$ of them, where $k$ is a parameter picked by the server at the start of the protocol, that is known to all the parties.

In parallel, $\mpcnodes$ also compute the norm of the model updates received, and set the bound to be $\lambda$ times the median of all norm values of the updates. Here $\lambda$ is a non-zero constant to control the tightness of the bound. The final updates to be aggregated are the ones that pass both the accuracy check, and the norm bound check.

$\mpcnodes$ then use $\PRec$ to reconstruct the aggregated update, $\bfu_{\msf{aggr}}$. The formal description of the protocol appears in \figref{pisecagg}. In order to use $\FMaxSoft$, we only need to replace the calls to $\FSecInf$ with $\FMaxSoft$, in \figref{pisecagg}.

\section{Security of $\PSecAgg$}
\label{app:proofs}

We prove our framework secure in the real-world/ideal-world simulation paradigm~\cite{EPRINT:Lindell16}. Security is argued by showing that whatever an adversary can do in the real world, it can do in the ideal world, where it interacts with an ideal functionality (trusted third party). In the ideal world, parties send their inputs to the ideal functionality, which then computes the desired function, and returns the output. In the real world, parties run the steps of the protocol.

The adversary is modeled as a probabilistic polynomial time (PPT) algorithm, that corrupts $t$, an honest minority of parties, among the parties involved in the protocol ($\PartySet{})$ in the real world.

We define an ideal functionality, $\FSecAgg$, as follows. It receives the models and the test data from the parties at the start of each round, computes the accuracy of the models on the test data, applies and the predicate $\valid()$ on them. It then filters out the bad predicates, aggregates the rest, and returns the output to the parties. The formal description appears in \figref{fsecagg}.

\begin{protofig}{Functionality $\FSecAgg$}{Functionality for Robust Secure Aggregation}{fsecagg}

    \textbf{Parameters:} Parties $\PartySet{} = \{P_1, \ldots, P_m\}$, number of corrupt parties $\cor$. Description of a check function, represented by $\msf{Chk}$. \\

    \textbf{Run:} On receiving $(\arith{\vec m}, \arith{\vec m'}, \arith{\vec D})$ from each party, where $\vec m$ is the updated local model, and $\vec m'$ is the model before the update, do the following:

    \begin{enumerate}
        \item Reconstruct the models $\vec m = \Sigma \arith{\vec m}$, $\vec m' = \Sigma \arith{\vec m'}$, and the test data $\vec D = \Sigma \arith{\vec D}$.
        \item Compute the inference of each of the models, $\msf{Acc}_i = \msf{Chk}(\vec m_i, \vec D)$, and $\msf{Acc'}_{i - 1} = \msf{Chk}(\vec m'_i, \vec D)$ for $i \in [1, m]$. Compute $\scr_i = \msf{Acc}_i - \msf{Acc}'_i$, for $i \in [1, m]$.
        \item Sort the vector of $\scr$ values, and apply a trimmed mean on them by deleting the top $\cor/2$ and bottom $\cor/2$ scores. Denote the final score for $\vec m_i$ is denoted as $\msf{scr}_i$.
        \item Select the models corresponding to the top $k$ scores across the $m$ parties.
        \item Compute the norms of these models, $\msf{norm}_i$, for $i \in [1, m - k]$. Set the value $\msf{bound}$ to be 1.5 times the median of the norms.
        \item Aggregate all models such that $\msf{norm}_i \leq \msf{bound}$ as $m_{\msf{agg}}$.
        \item Send $m_{\msf{agg}}$ to $\Adv$. If $\Adv$ sends $\msf{abort}$, send $\msf{abort}$ to all the parties. Else, send $m_{\msf{agg}}$.
    \end{enumerate}

    \textbf{Output to the adversary:} Send $\arith{m_{\msf{agg}}}$ to $\Adv$, and wait for it to respond. If $\Adv$ sends $\msf{abort}$, send $\msf{abort}$ to all the parties. Otherwise, send $\arith{m_{\msf{agg}}}$ to all the parties.

\end{protofig}

\begin{protofig}{Functionality $\FSort$}{Functionality for Oblivious Sort}{fsort}

\textbf{Parameters:} Parties $\PartySet{} = \{P_1, \ldots, P_m\}$.

\begin{enumerate}
    \item On receiving a sharing of a vector of elements, $\arith{\vec u}_i$, from all the parties reconstruct $\vec u = \Sigma_{i = 1}^m \arith{\vec u}_i$.
    \item Define the sorted array as $\vec u'$. Send $\arith{\vec u'}$ to all the parties.
\end{enumerate}
    
\end{protofig}

\begin{protofig}{Functionality $\FRand$}{Functionality to Sample Common Random Values}{frand}

\textbf{Parameters:} Parties $\PartySet{} = \{P_1, \ldots, P_m\}$. A pseudorandom function (PRF) $F$.

\begin{enumerate}
    \item Sample pairwise PRF keys, $\kappa_{ij}$ for each pair of parties $(P_i, P_j)$, for $i, j \in [1, m]$. Sample a common key $\kappa$ for all the parties.
    \item Send all $\kappa, \kappa_{ij}$, where $P_i$ is controlled by $\Adv$, to $\Adv$. If $\Adv$ sends $\msf{abort}$, send $\msf{abort}$ to all the parties. Else, send $\kappa$ and the respective keys to all the parties.
\end{enumerate}
    
\end{protofig}

\begin{protofig}{Functionality $\FSecInf$}{Functionality for Secure Inference}{fsecinf}

\textbf{Parameters:} Parties $\PartySet{} = \{P_1, \ldots, P_m\}$, a neural network architecture $N$.

\begin{enumerate}
    \item On receiving $(\arith{\vec w}, \arith{\vec D})$ from all the parties, reconstruct $\vec w = \Sigma_{i = 1}^m \arith{\vec w}$ and $\vec D = \Sigma_{i = 1}^m \arith{\vec D}$.
    \item Compute inference using the weights $\vec w$, on the dataset $D$, for a neural network $N$, to get $\msf{Acc}$.
    \item Send $\arith{\msf{Acc}}$ to $\Adv$ and wait for it to respond. If $\Adv$ sends $\msf{abort}$, send $\msf{abort}$ to all the parties. Else, send $\arith{\msf{Acc}}$.
\end{enumerate} 
    
\end{protofig}

\begin{protofig}{Functionality $\FMult$}{Functionality for Multiplication}{fmult}

\textbf{Parameters:} Parties $\PartySet{} = \{P_1, \ldots, P_m\}$.

\begin{enumerate}
    \item Receives secret-shared values $\arith{x}, \arith{y}$ from $\PartySet{}$, and reconstructs $x, y$.
    \item Computes $z = x \cdot y$, and sends it to $\Adv$.
    \item Depending on what $\Adv$ replies with, sends either $\arith{z}$ to the parties, or $\msf{abort}$.
\end{enumerate}

\end{protofig}

\begin{protofig}{Functionality $\FMaxSoft$}{Functionality for Maximum Softmax}{fmaxsoft}

\textbf{Parameters:} Parties $\PartySet{} = \{P_1, \ldots, P_m\}$, a classifier $f$, with $L$ classes.

\begin{enumerate}
    \item On receiving $(\arith{\vec w}, \arith{\vec D})$ from all the parties, reconstruct $\vec w = \Sigma_{i = 1}^m \arith{\vec w}$ and $\vec D = \Sigma_{i = 1}^m \arith{\vec D}$.
    \item Compute $\msf{MSft} = \frac{1}{|\bfD|} \sum_{(\bfx_k, y_k) \in \bfD} \max f(\bfx_k; \bfw)$, where $f(\bfx_k; \bfw) \in \R^L$ is the predicted softmax probability over $L$ classes, which is an output by the classifier $f$ parameterized with $\bfw$. 
    \item Send $\arith{\msf{MSft}}$ to $\Adv$ and wait for it to respond. If $\Adv$ sends $\msf{abort}$, send $\msf{abort}$ to all the parties. Else, send $\arith{\msf{MSft}}$.
\end{enumerate} 
    
\end{protofig}

\begin{theorem}\label{thm:secagg}
    Protocol $\PSecAgg$ securely realises the functionality $\FSecAgg$ in the presence of a malicious adversary that can statically corrupt up to $\cor < m/2$ parties in $\mpcnodes$, in the $(\FSecInf, \FSort, \FRand)$-hybrid model.
\end{theorem}

\begin{proof}
    Since the parties in the protocol only interact with $\FSecInf$, $\FMult$, $\FSort$, $\FZeroOne$, and $\FRand$ for the key setup, the proof is quite straightforward. Let $A$ denote the set of corrupt parties. We construct a simulator, $\Sim$ that interacts with the adversary controlled parties in the real world, and ideal functionality, $\FSecAgg$, in the ideal world. $\Sim$ initializes a boolean flag, $\msf{flag} = 0$, which indicates whether an honest party aborts during the protocol.

    Note that $\mpcnodes$ does not have any inputs, and only acts as output-receiving parties. The inputs, namely the models $\vec w_i$, and the validation datasets, $\vec D_i$, for $i \in [1, m]$, come from the clients. Thus, we can distinguish between two cases. The honest clients communicate their data directly to the functionality, $\FSecAgg$, and $\Sim$ does not need to simulate anything. If $\Adv$ sends an $\msf{abort}$ to it, it forwards that to $\FSecAgg$. Else, it receives $m_{\msf{agg}}$ from the functionality, and sends it to $\Adv$.

    The Setup phase is simulated by $\Sim$ by invoking the simulator for $\FRand$. In order for a corrupt client to secret-share its input, $\mpcnodes$ reconstruct a fresh random value $u$ to it, using their pairwise PRF keys.
    Since we assume that we are using a linear secret-sharing scheme, such as an instantiation of a replicated secret-sharing scheme~\cite{CCS:MohRin18,USENIX:DalEscKel21,tetrad,EPRINT:FLNW16}, the simulator knows $\Adv$'s keys, and can extract the corrupt client's input, $\vec w', \vec D'$. $\Sim$ runs the simulator for $\FSecInf$, by using random data for the cases when the $\vec w'$ is supposed to be checked with honest client's validation data. $\Sim$ emulates $\FSort$ internally.

    If the adversary cheated at any point in the protocol, set $\msf{flag} = 1$, and send $\msf{abort}$ to $\FSecAgg$. Else, send $\vec w', \vec D'$ to $\FSecAgg$, to receive $m_{\msf{agg}}$. If $\msf{flag} = 1$, or if $\Adv$ sends an $\msf{abort}$, send $\msf{abort}$ to $\FSecAgg$. Else, send $m_{\msf{agg}}$ to $\Adv$.

\end{proof}

\section{Experimental Details}
We repeat each setup for 10 trials (using random seeds 40-49) and report the average.
All the experiments are run on a server with thirty-two AMD EPYC 7313P 16-core processors, 528 GB of memory, and four Nvidia A100 GPUs. Each GPU has 80 GB of memory.

\subsection{Detailed Experimental Setup}
\label{app:setup}

\mypara{Datasets.}
We consider the following four datasets:
\begin{enumerate}
    \item MNIST~\cite{lecun2010mnist} is $28\times28$ grayscale image dataset for 0-9 digit classification. It consists of 60k training and 10k test images balanced over 10 classes.
    \item Fashion-MNIST (FMNIST)~\cite{fmnist} is identical to MNIST in terms of the image size and format, the number of classes, and the number of training and test images.
    \item SVHN~\cite{svhn} is a dataset for a more complicated digit classification benchmark than MNIST. It contains 73,257 training and 26,032 testing 32×32 RGB images of printed digits (from 0 to 9) cropped from real-world pictures of house number plates.

    \item CIFAR-10~\cite{cifar10} is $32\times32$ colored images balanced over ten object classes. It has 50k training and 10k test images.
\end{enumerate}
We follow the given train vs. test set split and further split each train set, balanced across classes. We follow the setup in~\cite{CCS:CGJv22} and reserve random 10k training samples as the public validation set ($\calD_{pubval}$) for aggregation baselines that leverage public validation datasets. For fair comparison, the remaining training samples are partitioned across clients by a Dirichlet distribution $\text{Dir}_K(\alpha)$ and used for training in all settings. We set $\alpha = 0.5$ to emulate data heterogeneity in realistic FL scenarios~\cite{cao2020fltrust}.
We assume there are 100 clients where $10\%$ or $20\%$ of them are malicious clients.

\mypara{Models.} Each client has a local model with one of the two following model architectures:
\begin{enumerate}
    \item LeNet-5~\cite{lecun1998gradient} with 5 convolutional layers and 60k parameters are used for MNIST and FMNIST.
    \item ResNet-18~\cite{resnet} with 18 layers and 11M parameters is used for SVHN and CIFAR-10.
\end{enumerate}

\mypara{Metrics.}
We compute the global model's accuracy on the entire test set at the end of each communication round. We repeat 10 trials using a random seed from 40 to 49 for each setting and report the average accuracy.

\mypara{Setup for distribution shifts in Section~\ref{sec: empirical-adaptability}.} We emulate the aforementioned scenarios using the MNIST data for existing clients and SVHN data for evolving or new clients, each with LeNet-5 as the model.
Note that both the original and shifted data share the same label set (\ie digits 0-9), simulating the covariate shifts~\cite{ood_bench, bai2023feed}. This ensures that a single global model will be capable of accommodating both distributions.
In the initial 250 rounds, the server communicates with 100 MNIST clients and establishes a public validation dataset using MNIST data, which remains static throughout the communication.
Following this, for Scenario 1, every 250 rounds, a random selection of 20 out of 100 clients begin to employ SVHN data and share local updates based on this data.
For Scenario 2, every 250 rounds, along with the 100 existing MNIST clients, a new set of 20 clients using randomly selected SVHN data join the communication. Consequently, the total number of participating clients increases to 120.
We transform the SVHN data into greyscale images, resized to 28x28 to match the format of the MNIST data. It is also split using the same non-IID split factor ($\alpha = 0.5$), ensuring consistency in the data setting for both MNIST and SVHN.
We report the global model's accuracy on the MNIST test set for the first 250 rounds and on the union of MNIST and SVHN tests for the remaining.
We also vary the periods of the distribution shifts by experimenting with intervals of 250 and 100 rounds.

\subsection{Soft $\score{}{}$ with Confidence Score}
\label{app:softscore}
We highlight that the computation of $\score{}{}$ is not limited to the accuracy difference in Equation~\ref{eqn:check_acc}. Rather,  one could compute other functions such as the maximum softmax value:
\begin{align}
\label{eqn:check_prob}
    &\score{i}{j} = \msf{Chk}_{prob}(\bfw_i, \bfD^{\msf{val}}_j) , \\ \nonumber
    &\msf{Chk}_{prob}(\bfw, \bfD) = \frac{1}{|\bfD|} \sum_{(\bfx_k, y_k) \in \bfD} \max_l ~ \bff(\bfx_k; \bfw)
\end{align}
where $\bff(\bfx_k; \bfw) \in \R^L$ is the predicted softmax probability over $L$ classes.
The above scoring function gauges the confidence of the model's predictions on the given validation data. Intuitively, given an uncorrupted set of data, which is the majority case, a benign model would have high confidence while a corrupted model would output low confidence on their predictions~\cite{hendrycks2022msp, hendrycks2022scaling}.

\subsection{Elaboration on Adaptive Attacks}
\label{app:adaptive-attack}
Assume the first $\cor$ out of $m$ clients are malicious without loss of generality. Then the adaptive local model poisoning attack aims to optimize 
\begin{align}    
\label{eqn:adaptive-attack}
&\max_{\bfw'_1,\cdots,\bfw'_{\cor}} \bfs^{T}(\bfw - \bfw')\\ \nonumber
\text{subject to} ~&\bfw = \calA(\bfw_1, \cdots, \bfw_{\cor}, \bfw_{\cor+1} \cdots, \bfw_{m})\\ \nonumber
&\bfw' = \calA(\bfw'_1, \cdots, \bfw'_{\cor}, \bfw_{\cor+1} \cdots, \bfw_{m})
\end{align}
where $\bfw_i$ is the clean local model under no attack from client $i$, $\bfw'_{i}$ is the compromised model from client $i$, and $\bfs$ represents the sign of change of clean global model in the current iteration under no attack.
According to~\cite{fang2020local} the solution can be found heuristically by solving
\begin{align}
\label{eqn:adaptive-attack-approx}
&\max_{\lambda} \ \ \  \lambda \\ \nonumber
\text{subject to} ~&\bfw'_1 = \bfw - \lambda\bfs, \\ \nonumber
&\bfw'_i = \bfw'_1 \mbox{ and } \pmb{\alpha}_i=1, \forall i\in [\cor]\\ \nonumber
\end{align}

In this heuristics, the malicious clients may collude to send the same poisoning model or support each other to be included in the final aggregation, which pull the aggregate model update in the opposite direction under no attack. The adversary wants to find the largest magnitude $\lambda$ of moving in the opposite direction such that the poisoning models $\bfw'_i$s are still included in the final aggregation.

\mypara{Adversary's Knowledge.}
In order to solve the optimization problem in Eqn.~\ref{eqn:adaptive-attack}, the adversary can make the malicious clients collude with each other and have access to their local training data, local model updates, loss function, and learning rate on the malicious clients, but not on the benign clients.
It also knows the entire defense mechanism and any thresholds used by the mechanism; \ie in our proposed aggregation protocol, we accept the top 90\% client models ranked by their validation scores (in the case of 10\% corruption). 
In RoFL, EIFFeL, or Elsa, for example, threshold parameters used in their validation predicate are also known as they are broadcast to everyone at the start of every round. On the contrary, the threshold values cannot be revealed to anyone in ours as they are hidden inside the MPC node.
Accordingly, they do not know $\bfs$, $\bfw$, and $\bfw’$ directly, but can only estimate them by colluding with each other. 
This corresponds to \textit{partial knowledge} in Section 3 of \citet{fang2020local}. The adversary approximates $\bfs$ using the before-attack local updates on the malicious clients at each communication round.

We also note that the adversary does \emph{not} know 1) the exact data used for validation, and 2) the exact integrity check result list $\calV$ because 1) is our privacy goal and 2) $\calV$ is only computed after the malicious clients stage their model update.

\mypara{Adaptive Attack to Norm bound and Norm ball.}
We can derive closed-form solutions from Eqn.~\ref{eqn:adaptive-attack} for norm bound and norm ball in the simplest case.
The objective in Eqn.~\ref{eqn:adaptive-attack} can be rewritten as follows: assuming a basic mean aggregation for $\calA$, and a robust validation check $\msf{Chk}(\bfw)$,
\begin{align*}
\bfs^T (\bfw - \bfw') 
  &= \bfs^T \Biggl\{ \left( \sum_{i=1}^{\cor} \bfw_i\, \indicator{\msf{Chk}(\bfw_i) = 1} 
           + \sum_{i=\cor+1}^{m} \bfw_i\, \indicator{\msf{Chk}(\bfw_i) = 1} \right) \\
  &\quad\; - \left( \sum_{i=1}^{\cor} \bfw'_i\, \indicator{\msf{Chk}(\bfw'_i) = 1} 
           + \sum_{i=\cor+1}^{m} \bfw_i\, \indicator{\msf{Chk}(\bfw_i) = 1} \right) \Biggr\} \\
  &= \bfs^T \left( \sum_{i=1}^{\cor} \bfw_i\, \indicator{\msf{Chk}(\bfw_i) = 1} 
           - \sum_{i=1}^{\cor} \bfw'_i\, \indicator{\msf{Chk}(\bfw'_i) = 1} \right)
\end{align*}
Since the first term is constant independent to $\bfw'_1, \dots, \bfw'_{\cor}$, now the objective becomes,
\begin{align*}
&\min_{\bfw'_1,\cdots,\bfw'_{\cor}} \bfs^{T} \sum_{i=1}^{\cor} \bfw'_i \indicator{\msf{Chk}(\bfw'_i) = 1} \\ \nonumber
\text{subject to} ~&\bfw = \calA(\bfw_1, \cdots, \bfw_{\cor}, \bfw_{\cor+1} \cdots, \bfw_{m})\\ \nonumber
&\bfw' = \calA(\bfw'_1, \cdots, \bfw'_{\cor}, \bfw_{\cor+1} \cdots, \bfw_{m})
\end{align*}

For norm bound, we further constrain that $||\bfw'_i||_2 < \tau, \forall i \in [\cor]$ so that malicious updates always pass the validity check and are included in the aggregation.
Then we obtain $\bfw'_i = - \bfs \cdot \frac{\tau - \epsilon}{\sqrt{d}}$ where $||\bfs||_2 = \sqrt{d}$ and $\epsilon$ is some positive small value (\eg $10^{-6}$) to ensure the inequality constraint.

For norm ball, we further constrain that $||\bfw'_i - \bfw_{pubval}|| < \tau, \forall i \in [\cor]$. 
We obtain $\bfw'_i = \bfw_{pubval} + \Delta \bfw'_i$ where $\Delta \bfw'_i = -\bfs \cdot \min (\tau - \epsilon, \frac{\tau - \epsilon}{\sqrt{d}})$.

\mypara{Adaptive Attack to Cosine similarity.}
We follow the instantiation of adaptive attack in \citet{fltrust} under partial knowledge setup.

\mypara{Adaptive Attack to \ours.}
To check if the poisoning models $\bfw'_{i}$s are still included in the aggregate, the adversary will use its own dataset to estimate the accuracy of its poisoning models as well as the benign models. Let $\bfD'$ denote the dataset stored at the malicious clients. The adversary computes the accuracy of all benign models $\bfw_{\cor+1}, \cdots, \bfw_{m}$ on $\bfD'$. Then for each $\lambda$ and its associated poisoning model $\bfw_i' = \bfw - \lambda\bfs$, the adversary evaluate its accuracy on $\bfD'$ and only accept $\lambda$ if the accuracy enters the top 50\% among the accuracy of benign models. We sample $\lambda\in\{1e^{-5}, 1e^{-1}\}$ and choose the largest accepted $\lambda$ to construct poisoning models.

\subsection{Proof: Optimality of Extreme Manipulation}
\label{app:extrememanipultationproof}

\begin{definition}[Extreme Check Score Manipulation]
\label{def:optimalattack}
When $\calC_{j}$ is malicious, a check score manipulation is \emph{extreme} if
\begin{equation}
\label{eqn:extreme}
\score{i}{j} =
\begin{cases}
1 &\mbox{ if $\calC_i$ is malicious,}\\
-1 &\mbox{ otherwise.}
\end{cases}
\end{equation}
\end{definition}

We show in Theorem~\ref{thm:optimal} the optimality of the check score manipulation in Definition~\ref{def:optimalattack} in the sense that such manipulation will also allow some malicious model updates to maximize the adversarial objective. 
\begin{theorem}(Optimality of the Extreme Manipulation.)
\label{thm:optimal}
    Let $\bfw^{(T)}_{adv}$ be a model that maximizes the objective in Formula~\ref{obj:adaptiveadv} and is obtainable by some combination of malicious model updates and validation data set manipulation. Then there must be a set of malicious model updates that lead to $\bfw^{(T)}_{adv}$ under the extreme check score manipulation in Definition~\ref{def:optimalattack}.
\end{theorem}

\begin{proof}
    First, we observe that the range of $\scr(i,j)$ in our protocol is from $-1$ to $1$.
    Consider a malicious $\calC_i$. When $C_j$ is malicious, $\scr(i, j)$ equals 1, which is the largest possible. When $C_j$ is benign, $\scr(i, j)$ is independent to malicious data manipulation. As a result, $\scr_i$ for a malicious client $\calC_i$ under the extreme manipulation is no less than under other manipulation. Similarly, $\scr_i$ for a benign client under the extreme manipulation will be no more than that under other manipulation. Therefore, the ranking of $\scr_i$ for a malicious client $\calC_i$ will never drop when switching from a non-extreme manipulation to extreme manipulation.

    Let $\bfw_{adv}$ be a malicious model update from $\calC_i$ in the sequence that leads to $\bfw^{(T)}_{adv}$ under some non-extreme validation data manipulation.
    \begin{enumerate}
        \item
        If $\bfw_{adv}$ is included in aggregation, then $\scr_i$ must already be in the top $k\%$ among all clients, $k$ being the threshold defined in the protocol. Since switching to the extreme manipulation will never reduce the ranking of a malicious $\scr_i$, $\bfw_{adv}$ will remain in aggregation. 
        \item
        If $\bfw_{adv}$ is not included in aggregation, then the attacker can send any $\bfw'_{adv}$ that will be excluded from aggregation as well.~\footnote{If such a $\bfw'_{adv}$ does not exist, i.e., all updates from the malicious client will be accepted, then we allow the attacker to refrain from submitting the updates. The extreme manipulation serves to check the robustness lower bound of our work. Therefore, we allow it more flexibility, i.e. more attacker strength.}
    \end{enumerate}
    Therefore, an adversary using extreme manipulation can always construct malicious model updates that have the same impact in each aggregation step, and thus obtain $\bfw^{(T)}_{adv}$.
\end{proof}

\mypara{Robustness Lower Bound.} We acknowledge that in practice, the adversary may not find an actual validation set that achieves the conditions in Equation~\ref{eqn:extreme}. Nevertheless, we equip the adaptive attack against our framework with the extreme check score manipulation in our empirical study in Section~\ref{sec:experiment}.
We implement the extreme manipulation by allowing the adversary to directly set the value of $\scr(i, j)$ when $\calC_j$ is a malicious client. By giving this additional power to the adversary, we can find an empirical \emph{lower} bound to the robustness of our work. 
We also note that it is still possible for an attacker to craft validation sets to boost scores for malicious clients while lowering those of benign ones in reality. 
They may flip labels to disrupt benign updates or embed backdoors that trigger high scores for malicious updates.
Thus, our empirical lower bound serves as a practical robustness indicator.

\section{MPC Benchmark of \ours}
\label{app:mpc-benchmark}

\mypara{Setup.} The framework is instantiated using building blocks from MP-SPDZ~\cite{CCS:Keller20}. Benchmarks were conducted on a MacBook Pro with an Apple Silicon M4 Pro processor and 48 GB RAM. We simulated the network setup locally, and ran all the parties on the same machine, with 4 threads. For the LAN case, we simulated a network with 1 Gbps bandwidth, and 1 ms latency, and for WAN, we used 200 Mbps and 20 ms respectively.

We benchmark the online phases of three components of our MPC protocol -- the norm bound check, the cross-client check, and selecting the top-k gradients separately. In reality, the normbound and cross-client check can be run in parallel, with top-k run on the outputs of these checks. All the components are instantiated using a semi-honest 3PC protocol from MP-SPDZ, specifically, replicated-ring-party.x which implements the replicated secret-sharing based protocol from \cite{CCS:AFLNO16}. The protocols are run with a 128-bit ring, although the larger ring size compared to the typical 64-bit ring is only to the size of the norm. It is possible to operate over a 64-bit ring, and switch to a 128-bit ring only for the norm, although we do not implement this optimization. 

For most cases, max-prob combined with normbound gives a good balance between robustness and runtime efficiency. However, in some scenarios the full accuracy check provides a nontrivial amount of robustness gain, which is why we also report the times for it. Note that normbound is generally faster than both versions of the cross-client checks, except in the case of ResNet-18. This is because the vector for which the norm is computed is of size 11,173,962 in the case of ResNet-18, and due to a limitation with the amount of RAM available, we could not take full advantage of parallelism.

\end{document}